\documentclass[sigplan,9pt]{acmart}
\newcommand{\vlong}[1]{#1}
\newcommand{\vshort}[1]{}

\usepackage[utf8]{inputenc}
\usepackage{amsmath,amsthm,amssymb}
\usepackage{mathrsfs}

\usepackage{bm}
\usepackage{xparse}
\usepackage{url}
\usepackage{stmaryrd}
\usepackage{textcomp}
\usepackage{proof}
\usepackage{color}
\usepackage{float}
\usepackage{tikz}
\usetikzlibrary{arrows,positioning,decorations.pathmorphing,arrows.meta,matrix,svg.path}
\usepackage{cmll}
\usepackage{subscript}

\definecolor{mygreen}{rgb}{0,0.6,0.4}
\definecolor{myred}{rgb}{0.8,0.1,0.1}
\definecolor{myblue}{rgb}{0,0.5,0.8}

\usepackage{listings}
\usepackage{doi}
\usepackage{graphicx}
\usepackage{subcaption}
\usepackage{diagbox}
\usepackage{xfrac}
\usepackage{cancel} 
\usepackage{cleveref}
\usepackage[colorinlistoftodos,prependcaption,textsize=scriptsize,textwidth=2cm]{todonotes}

\makeatletter
\newcommand{\LeftEqNo}{\let\veqno\@@leqno}
\makeatother

\usepackage[USenglish]{babel}
\usepackage{multicol}
\usepackage{enumitem}
\setlist[itemize]{itemsep=0pt plus 1pt}
\setlist[enumerate]{itemsep=0pt plus 1pt}

\usepackage[strict]{changepage}
\usepackage{array}
\usepackage{ltxcmds}



\newcommand{\automath}[1]{\relax\ifmmode{#1}\else{$#1$}\fi}
\newcommand{\eqtau}{\equiv_\tau}

\theoremstyle{plain}
\makeatletter
\newtheorem*{rep@theorem}{\rep@title}
\newcommand{\newreptheorem}[2]{%
\newenvironment{rep#1}[1]{%
 \def\rep@title{#2 \ref{##1}}%
 \begin{rep@theorem}}%
 {\end{rep@theorem}}}
\makeatother
\newreptheorem{theorem}{Theorem}
\newreptheorem{lemma}{Lemma}
\newreptheorem{proposition}{Proposition}
\newreptheorem{corollary}{Corollary}

\newcommand{\dpaw }{\textrm{\lowercase{d}PA\textsuperscript{$\omega$}}}
\newcommand{\dlpaw}{\textrm{\lowercase{d}LPA\textsuperscript{$\omega$}}}

\newcommand{\nef}{\textsc{nef}}
\newcommand{\cps}{\textsc{CPS}}
\newcommand{\nth}{\texttt{nth}}

\newcommand{\prf}{\mathop{\texttt{prf}}\,}
\newcommand{\wit}{\mathop{\texttt{wit}}\,}
\newcommand{\refl}{\mathop{\texttt{refl}}}
\newcommand{\subst}[2]{\mathop{\texttt{subst}}\, #1\,#2}
\newcommand{\splitop}{\mathop{\texttt{split}}}
\renewcommand{\split}[2]{\splitop\, #1\, \mathop{\texttt{as}}\,(#2)\,\mathop{\texttt{in}}\,}
\newcommand{\case}[3]{\mathop{\texttt{case}} \,#1\, \mathop{\texttt{of}}\, [#2\,|\, #3]}
\newcommand{\exf}{\mathop{\texttt{exfalso}}\,}
\newcommand{\dest}[2]{\mathop{\texttt{dest}} \,#1\, \mathop{\texttt{as}}\, (#2)\, \mathop\texttt{in}\,}
\newcommand{\catch}[1]{\mathop{\texttt{catch}}_{#1}}
\newcommand{\throw}[1]{\mathop{\texttt{throw}}\,{#1}}
\newcommand{\letop}{\mathop{\texttt{let}}}
\newcommand{\inop}{\mathop{\texttt{in}}}
\newcommand{\letin}[2]{\letop #1=#2 \inop}
\newcommand{\ind}[4]{\texttt{fix}^{#1}_{#3} [#2\,|\,#4]}
\newcommand{\cofix}[3]{{\texttt{cofix}}^{#1}_{#2}{[#3]}}

\newcommand{\recop}{\mathop{\texttt{rec}}}
\newcommand{\rec}[4]{{\recop}^{#1}_{#2}[#3\,|\,#4]}

\newcommand{\injec}[2]{\iota_{#1}(#2)}

\newcommand{\mut}{\tilde{\mu}}

\newcommand{\coupesep}{|\!|}
\newcommand{\coupe}[2]{
    {\mbox{$\langle$}} {#1} {\coupesep} {#2} \mbox{$\rangle$}}

\newcommand{\typered}{\triangleright}
\newcommand{\reds}{\rightsquigarrow_s}

\newcommand{\dltp}{\textrm{\normalfont{d{L}\textsubscript{$\reset$}}}}
\newcommand{\tp}{{{\texttt{t}\!\texttt{p}}}}
\newcommand{\reset}{{\hat{\tp}}}
\newcommand{\coreset}{{\check{\tp}}}
\renewcommand{\shift}{\mu{\reset}.}
\newcommand{\coshift}{\tmu{\coreset}.}

\makeatletter
\newcommand{\ovset}[3][0ex]{%
  \mathrel{\mathop{#3}\limits^{
    \vbox to #1{\kern-2\ex@
    \hbox{#2}\vss}}}}
\makeatother

\newcommand{\red}{\rightarrow}  
\newcommand{\basearrow}{\longrightarrow}

\newcommand{\bmark}{\beta     }

\newcommand{\bred} {\basearrow_{\bmark}}

\makeatletter
\newcommand{\shorteq}{%
  \settowidth{\@tempdima}{-}
  \resizebox{\@tempdima}{\height}{=}%
}
\makeatother

\newcommand{\tmu}{{\tilde\mu}}
\newcommand{\cut}[2]{\coupe{#1}{#2}}
\newcommand{\cmd}[2]{\coupe{#1}{#2}}
\newcommand{\imp}{\rightarrow}
\newcommand{\defeq}{\triangleq}
\newcommand{\lbvtstar}{{\overline{\lambda}_{[lv\tau\star]}}}

\newcommand{\pole}{{\bot\!\!\!\bot}}
\newcommand{\orth}{{\pole}}

\newcommand{\dpt}[1]{\{#1\}}
\newcommand{\eps}{\varepsilon}
\newcommand{\rdpt}[1]{\dpt{\cdot|#1}}

\newcommand{\negt}{{\bot\!\!\!\bot}}
\newcommand{\dptimp}{\Rrightarrow}
\newcommand{\dptprod}[2]{\Pi{#1}.{#2}}
\newcommand{\mustar}{\mu\!\star\!}

\newcommand{\N}{{\mathbb{N}}}

\renewcommand{\P}{\mathcal{P}}

\newcommand{\lmmt}{$\lambda\mu\tmu$}

\newcommand{\fvn}[2]{\|#2\|_{#1}}
\newcommand{\tvn}[2]{|#2|_{#1}}

\newcommand{\fvf}[1]{\fvn{f}{#1}}

\newcommand{\fve}[1]{\fvn{e}{#1}}
\newcommand{\tvv}[1]{\tvn{v}{#1}}
\newcommand{\tvV}[1]{\tvn{V}{#1}}
\newcommand{\tvVt}[1]{\tvn{V_t}{#1}}
\newcommand{\fvpi}[1]{\tvn{\pi}{#1}}
\newcommand{\tvt}[1]{\tvn{t}{#1}}
\newcommand{\tvp}[1]{\tvn{p}{#1}}
\newcommand{\real}{\Vdash}

\newcommand{\prfcase}[1] {\paragraph{\normalfont\textbullet~ \textbf{Case}  #1.}}

\newcommand{\cmdt}[2]{\cmd{#1}{#2}_t}

\newcommand{\cmdp}[2]{\cmd{#1}{#2}_p}
\newcommand{\cmde}[2]{\cmd{#1}{#2}_e}
\newcommand{\cmdv}[2]{\cmd{#1}{#2}_v}
\newcommand{\cmdf}[2]{\cmd{#1}{#2}_f}

\newcommand{\cmdpi}[2]{\cmd{#1}{#2}_{\pi}}

\def\limp{\Rightarrow}

\def\<{\langle}
\def\>{\rangle}

\newcommand{\Quote}{\texttt{quote}}

\newcommand{\dom}{\mathop{\mathrm{dom}}}

\newcommand{\FV}{{FV}}
\renewcommand{\P}{\mathcal{P}}

\renewcommand{\C}{\mathcal{C}}

\renewcommand{\L}{\mathscr{L}}

\newcommand{\muteq}{\tmu\shorteq}

\newlength{\saut}

\newcommand{\myfig}[1]{\framebox{\vbox{#1}}}

\newcommand{\myquote}[1]{\emph}
\newcommand{\noem}{\vspace{-1em}}
\newcommand{\nomidem}{\vspace{-0.5em}}

\newcommand{\hole}{{[\,\,]}}

\newcommand{\sigdash      }{\vdash^{\sigma}} 
\newcommand{\sigsigdash      }{\vdash^{\sigma\sigma'}} 
\newcommand{\optsigma      }{} 
\newcommand{\sigmaopt      }{} 
\newcommand{\sigsigprimeopt}{} 
\newcommand{\tis}[2]{(#1|#2)}

\newcommand{\compat}[2]{#1\diamond#2}

\makeatletter
\newcommand{\osetb}[3][0ex]{%
  \mathrel{\mathop{#3}\limits^{
    \vbox to#1{\kern-2\ex@
    \hbox{#2}\vss}}}}
\makeatother

\newcommand{\indpt}[2]{{#1\##2}}
\newcommand{\stext}{\vartriangleleft}
\newcommand{\instore}[1]{#1-in-store}
\newcommand{\ct }{\instore{term}}

\newcommand{\ce }{\instore{context}}

\newcommand{\cp }{\instore{proof}}

\newcommand{\stjoin}[2]{\mathsf{join}(#1,#2)}

\newcommand{\str}[2]{\texttt{str}_\infty^{#1}\,#2}
\newcommand{\defrule}{\autorule{\texttt{def}}}
\newcommand{\Ainf}[1]{A_\infty^{#1}}
\newcommand{\elimmark}{E}

\newcommand{\autorule}[1]{\relax\ifmmode{\scriptstyle(#1)}\else$(#1)$\fi}

\newcommand{\cutrule}{\autorule{\textsc{Cut} }}

\newcommand{\axnrule}{\autorule{\textsc{Ax}_n}}
\newcommand{\axxrule}{\autorule{\textsc{Ax}_t}}

\newcommand{\axrrule}{\autorule{\textsc{Ax}_r}}
\newcommand{\axlrule}{\autorule{\textsc{Ax}_l}}

\newcommand{\murule }{\autorule{\mu          }}
\newcommand{\mutrule}{\autorule{\mut         }}

\newcommand{\botrule}{\autorule{\bot         }}

\newcommand{\imprrule}{\autorule{\imp_r         }}
\newcommand{\implrule}{\autorule{\imp_l         }}

\newcommand{\falrule}{\autorule{\forall_l         }}
\newcommand{\farrule}{\autorule{\forall_r         }}

\newcommand{\indrule}{\autorule{\texttt{fix}}}
\newcommand{\cofixrule}{\autorule{\texttt{cofix}}}

\newcommand{\exrrule}{\autorule{\exists_r         }}
\newcommand{\exlrule}{\autorule{\exists_l         }}

\newcommand{\andeurule}{\autorule{\land^1_\elimmark}}
\newcommand{\andedrule}{\autorule{\land^2_\elimmark}}
\newcommand{\andrrule }{\autorule{\land_r          }}
\newcommand{\andlrule }{\autorule{\land_l          }}

\newcommand{\orrrule }{\autorule{\lor_r           }}

\newcommand{\orlrule }{\autorule{\lor_l           }}

\newcommand{\weakrule  }{\autorule{w        }}
\newcommand{\weakdrule }{\autorule{w^d      }}

\newcommand{\prfrule  }{\autorule{\prf    }}

\newcommand{\convrrule}{\autorule{\equiv_r}}
\newcommand{\convlrule}{\autorule{\equiv_l}}

\newcommand{\reflrule }{\autorule{\refl}}
\newcommand{\witrule  }{\autorule{\wit }}

\newcommand{\lamrule}{\autorule{\lambda}}
\newcommand{\apprule}{\autorule{@      }}

\newcommand{\lrule}    {\autorule{l     }}
\newcommand{\dlrule}   {\autorule{l^d   }}

\newcommand{\tauprule   }{\autorule{\tau_p      }}
\newcommand{\tauerule   }{\autorule{\tau_e      }}

\newcommand{\taucatrule }{\autorule{\tau\tau'   }}

\newcommand{\resetrule}{\autorule{\reset}}
\newcommand{\shiftrule}{\autorule{\mu\reset}}
\newcommand{\dcutrule  }{\autorule{\textsc{Cut}^d}}
\newcommand{\dtcutrule  }{\autorule{\textsc{Cut}^d_t}}
\newcommand{\dmutrule  }{\autorule{{\mut}^d}}
\newcommand{\dxmutrule  }{\autorule{{\mut}^d_x}}
\newcommand{\eqrule}{\autorule{=_l}}

\newcommand{\splitrule    }{\autorule{ \splitop}        }

\newcommand{\dexlrule}{\autorule{\exists_l^d         }}
\newcommand{\dorlrule }{\autorule{\lor_l  ^d         }}
\newcommand{\dandlrule }{\autorule{\land_l^d          }}

\vshort{\copyrightyear{2018}
\acmYear{2018}
\setcopyright{acmlicensed}
\acmConference[LICS '18]{LICS '18: 33rd Annual ACM/IEEE Symposium on Logic in Computer Science}{July 9--12, 2018}{Oxford, United Kingdom}
\acmBooktitle{LICS '18: LICS '18: 33rd Annual ACM/IEEE Symposium on Logic in Computer Science, July 9--12, 2018, Oxford, United Kingdom}
\acmPrice{15.00}
\acmDOI{10.1145/3209108.3209199}
\acmISBN{978-1-4503-5583-4/18/07}}

\startPage{1}

\vlong{\setcopyright{none}
\acmYear{\emph{This is an extended version of LICS 2018 paper.}}
\acmConference{\vshort}{\vshort}{\noem}
\acmBooktitle{}
\acmPrice{}
\acmDOI{}
\acmISBN{}
}

\usepackage{booktabs}   
\usepackage{subcaption} 

\begin{document}

\title{A sequent calculus with dependent types for classical arithmetic}         


\author{Étienne Miquey}
\affiliation{
  \department{Équipe Gallinette}              
  \institution{INRIA, LS2N}            
  \streetaddress{2 Chemin de la Houssinière}
  \city{Nantes}
  \postcode{44000}
  \country{France}                    
}
\email{emiquey@inria.fr}          


\begin{abstract}
In a recent paper~\cite{Herbelin12}, 
Herbelin developed {\dpaw}, a calculus in which constructive proofs for the axioms of countable and dependent
choices could be derived via the encoding of a proof of countable universal quantification as a stream of it components.
However, the property of normalization (and therefore the one of soundness) was only conjectured.
The difficulty for the proof of normalization is due to the simultaneous presence 
of dependent types (for the constructive part of the choice),
of control operators (for classical logic), of coinductive objects (to encode
functions of type $\N\to A$ into streams $(a_0,a_1,\ldots)$) 
and of lazy evaluation with sharing (for these coinductive objects).

Elaborating on previous works, we introduce in this paper a variant of {\dpaw} presented as a sequent calculus.
On the one hand, we take advantage of a variant of Krivine classical realizability
that we developed to prove the normalization of classical call-by-need~\cite{MiqHer18}.
On the other hand, we benefit from \dltp, a classical sequent calculus with dependent types 
in which type safety is ensured by using delimited continuations together with a syntactic restriction~\cite{Miquey17}.
By combining the techniques developed in these papers,
we manage to define a realizability interpretation \emph{à la} Krivine of our calculus
that allows us to prove normalization and soundness.
\end{abstract}


\maketitle

\newcommand{\mypar}[1]{\subsection{#1}}
\section{Introduction}
\mypar{Realizing AC$_\N$ and DC in presence of classical logic}

Dependent types are one of the key features of Martin-L\"of's type theory~\cite{MartinLof98},
allowing formulas to refer to terms. 
Notably, the existential quantification rule is defined so that 
a proof term of type $\exists x^A\!.B$ is a pair $(t, p)$ where $t$---the \emph{witness}---is
of type $A$, while $p$---the \emph{proof}---is of type $B[t/x]$. 
Dually, the theory enjoys two elimination rules: one with a destructor $\wit$ to extract the witness,
the second one with a destructor $\prf$
to extract the proof.
This allows for a simple and constructive proof of the 
full axiom of choice~\cite{MartinLof98}:
$$
\begin{array}{rll}
 AC_A & := & \lambda H.(\lambda x.\wit(H x),\lambda x.\prf(H x)) \\
 &:& (\forall x^A\!.\exists y^B\!. P(x,y)) \imp \exists f^{A\imp B}\!.\forall x^A\!. P(x,f(x))
\end{array}
$$
This term is nothing more that an implementation of Brouwer-Heyting-Kolomogoroff interpretation
of the axiom of choice~\cite{Kolmogoroff1932}:
given a proof $H$ of $\forall x^A\!.\exists y^B\!. P(x,y)$, 
it constructs a choice function which simply maps any $x$ to the witness of $Hx$, while the proof that this function
is sound w.r.t. $P$ returns the corresponding certificate.

Yet, this approach deeply relies on the constructivity of the theory.
We present here a continuation of Herbelin's works~\cite{Herbelin12}, 
who proposed a way of scaling up Martin-L\"of's proof to classical logic.
The first idea is to restrict the dependent types to the fragment 
of \emph{negative-elimination-free} proofs (\nef) 
which, intuitively, only contains constructive proofs behaving as values.
The second idea is to represent a countable universal quantification as an 
infinite conjunction. This allows us to internalize into a formal system (called
\dpaw) the realizability approach~\cite{BerBezCoq98,EscOli14} as a direct 
proofs-as-programs interpretation.

Informally, let us imagine that given a proof $H:\forall x^\N\!.\exists y^B\!. P(x,y)$,
we could create the infinite sequence $H_\infty=(H 0,H 1,\ldots)$ 
and select its $n^{\textrm{th}}$-element with some function \nth.
Then, one might wish that: 
$$\lambda H.(\lambda n.\wit(\nth~n~H_\infty),\lambda n.\prf(\nth~n~H_\infty))$$
could stand for a proof for $AC_\N$.
One problem is that even if we were effectively able to build such a term, $H_\infty$ might
still contain some classical proofs. Therefore, two copies of $H n$ might end up behaving 
differently according to the contexts in which they are executed, and thus returning
two different witnesses (which is known to lead to logical inconsistencies~\cite{Herbelin05}). 
This problem can be fixed by using a shared version of $H_\infty$, that is to say: 
$$\lambda H.\letin{a}{H_\infty}~(\lambda n.\wit(\nth~n~a),\lambda n.\prf(\nth~n~a)\,.$$
In words, the term $H_\infty$ is now shared between all the places which may require
some of its components.

It only remains to formalize the intuition of $H_\infty$,
which is done by means of a stream $\cofix{0}{fn}{(H n, f(S(n)))}$ iterated on $f$ with 
parameter $n$, starting with 0:
$$\begin{array}{r@{}l}
AC_\N:=\lambda H.&\letop a=\cofix{0}{fn}{(H n, f(S(n))}\\
		 &\inop\, (\lambda n.\wit(\nth~n~a),\lambda n.\prf(\nth~n~a)\,.
\end{array}$$
The stream is, at the level of formulas, an inhabitant of a coinductively 
defined infinite conjunction $\nu^0_{Xn}(\exists y.P(n,y))\wedge X(n+1)$.
Since we cannot afford to pre-evaluate each of its components, and we thus have to use a 
\emph{lazy} call-by-value evaluation discipline.
However, it still might be responsible for some non-terminating reductions,
all the more as classical proofs may contain backtrack.

\mypar{Normalization of {\textbf{dPA}$^{\omega}$}}
In~\cite{Herbelin12}, the property of normalization (on which relies the one of consistency)
was only conjectured, and the proof sketch that was given turned out to be hard to formalize properly.
Our first attempt to prove the normalization of {\dpaw} was to derive a continuation-passing style translation (CPS),
but translations appeared to be hard to obtain for \dpaw~as such. 
In addition to the difficulties caused by control operators and co-fixpoints, 
{\dpaw} reduction system is defined in a natural deduction fashion, with contextual rules 
where the contexts involved can be of arbitrary depth. 
This kind of rules are indeed 
difficult to faithfully translate through a CPS. 

Rather than directly proving the normalization of {\dpaw}, we choose to first 
give an alternative presentation of the system under the form of a sequent calculus, which we call {\dlpaw}.
Indeed, sequent calculus presentations of a calculus usually provides good intermediate steps 
for {\cps} translations~\cite{Munch13PhD,MunSch15,AriDowMauJon16} since they enforce a decomposition of the reduction system
into finer-grain rules.
To this aim, we first handled separately the difficulties peculiar to the definition of such a calculus:
on the one hand, we proved with Herbelin the normalization of a calculus with control operators and lazy evaluation~\cite{MiqHer18};
on the other hand, we defined a classical sequent calculus with dependent types~\cite{Miquey17}.
By combining the techniques developed in these frameworks, we finally manage to define {\dlpaw}, 
which we present here and prove to be normalizing.

\mypar{Realizability interpretation of classical call-by-need}
In the call-by-need evaluation strategy, the substitution of a variable 
is delayed until knowing whether the argument is needed.
To this end, Ariola \emph{et al.}~\cite{AriEtAl12} proposed the $\lbvtstar$-calculus, 
a variant of Curien-Herbelin's {\lmmt}-calculus~\cite{CurHer00}
in which substitutions are stored in an explicit environment.
Thanks to Danvy's methodology of semantics artifacts~\cite{DanEtAl10},
which consists in successively refining the reduction system until getting 
context-free reduction rules\footnote{That is to say reduction rules in an abstract machine
for which only the term or the context needs to be analyzed in order to decide whether the rule can be applied.},
they obtained an untyped CPS translation for the $\lbvtstar$-calculus.
%
%
By pushing one step further this methodology, 
we showed with Herbelin how to obtain a realizability interpretation \emph{à la}
Krivine for this framework~\cite{MiqHer18}.
The main idea, in contrast to usual models of Krivine realizability~\cite{Krivine09},
is that realizers are defined as pairs of a term and a substitution.
The adequacy of the interpretation directly provided us with a proof of normalization,
and we shall follow here the same methodology to prove the normalization of~\dlpaw.

\mypar{A sequent calculus with dependent types}
While sequent calculi are naturally tailored to smoothly support
CPS interpretations, there was no such presentation of language with dependent types
compatible with a CPS.
In addition to the problem of safely combining control operators and dependent types~\cite{Herbelin05},
the presentation of a dependently typed language under the form of a sequent calculus 
is a challenge in itself. 
In~\cite{Miquey17}, we introduced such a system, called \dltp, 
which is a call-by-value sequent calculus with classical control  
and dependent types. 
In comparison with usual type systems, we decorated typing derivations with a list of dependencies
to ensure subject reduction. 
Besides, the soundness of the calculus was justified 
by means of a CPS translation taking the dependencies into account.
The very definition of the translation constrained us to use delimited continuations in the calculus
when reducing dependently typed terms. 
At the same time, this unveiled the need for the syntactic restriction of dependencies 
to the \emph{negative-elimination-free} fragment as in~\cite{Herbelin12}.
Additionally, we showed how to relate our calculus to a similar system by Lepigre~\cite{Lepigre16}, 
whose consistency is proved by means of a realizability interpretation.
In the present paper, we  use the same techniques, namely a list of dependencies and delimited continuations,
to ensure the soundness of {\dlpaw}, and we follow Lepigre's interpretation of dependent types for the definition
of our realizability model.

\mypar{Contributions of the paper}
The main contributions of this paper can be stated as follows.
First, we define {\dlpaw} (\Cref{s:dlpaw}), a sequent calculus with classical control, dependent types,
inductive and coinductive fixpoints and lazy evaluation made available thanks to 
the presence of stores. This calculus can be seen as a sound combination of \dltp~\cite{Miquey17}\vspace{-0.9mm} 
and the $\lbvtstar$-calculus~\cite{AriEtAl12,MiqHer18} extended with the expressive power of {\dpaw}~\cite{Herbelin12}.
Second, we prove the properties of normalization and soundness for {\dlpaw} thanks to a realizability
interpretation \emph{à la} Krivine, which we obtain by applying Danvy's methodology of semantic artifacts (\Cref{s:small_step,s:realizability}). 
Lastly, {\dlpaw} incidentally provides us with a direct proofs-as-programs 
interpretation of classical arithmetic with dependent choice, as sketched in~\cite{Herbelin12}.

\paragraph{}\emph{This paper is partially taken from the Chapter 8 of the author's PhD thesis~\cite{these}. 
For more detailed proofs, we refer the reader to the \vlong{different} appendices\vshort{ of the version available at: \url{https://hal.inria.fr/hal-01703526}}.}

\section{A sequent calculus with dependent types for classical arithmetic}
\label{s:dlpaw}
\subsection{Syntax}
\begin{figure*}[t]
  \framebox{\vbox{
\begin{tabular}{>{}l|l}
 $\begin{array}{@{}l}
  \textbf{Closures}\\ 
  \textbf{Commands}\\[0.8em] 
  \textbf{Proof terms }\\  \\
  \textbf{Proof values} \\[0.8em]
  \textbf{Terms}\\ 
  \textbf{Terms values} \\\\
  \end{array}~
\begin{array}{c@{~}c@{~}l}
l	& ::= & c\tau \\
c	& ::= & \cmd{p}{e} \\[0.8em]

p,q 	& ::= & a \mid \injec{i}{p} \mid (p,q) \mid (t,p) \mid \lambda x.p \mid \lambda a.p\mid \refl \\
	& \mid&  \ind{t}{p_0}{ax}{p_S} \mid \cofix{t}{bx}{p}\mid \mu \alpha.c \mid \shift {c_\reset}\\
V	& ::= & a  \mid \injec{i}{V} \mid (V,V) \mid (V_t,V) \mid \lambda x.p \mid \lambda a.p \mid \refl \\[0.8em]

t,u 	& ::= & x \mid 0 \mid S(t) \mid \rec{t}{xy}{t_0}{t_S}\mid \lambda x.t \mid t~u \mid \wit p \\
V_t 	& ::= & x \mid S^n(0) \mid  \lambda x.t\\\\
       
\end{array}$
&
$
\begin{array}{@{}l}
  \textbf{Stores} \\
  \textbf{Storables} \\ [0.8em]
  \textbf{Contexts}\\
  \textbf{Forcing}\\\textbf{contexts}\\  [0.8em]
  \textbf{Delimited}\\ 
  \textbf{continuations}\\\\
  \end{array}
\begin{array}{c@{~}c@{~}l}
\tau 	& ::= & \varepsilon \mid \tau [a := p_\tau] \mid \tau[\alpha := e]\\ 
p_\tau 	& ::= & V \mid \ind{V_t}{p_0}{ax}{p_S} \mid \cofix{V_t}{bx}{p}\\
[0.8em]
	    
e	& ::= & f \mid  \alpha \mid \mut a.c\tau \\ 
f	& ::= & [] \mid \mut[a_1.c_1\mid a_2.c_2]  \mid \mut(a_1,a_2).c  \\
	& \mid&  \mut(x,a).c \mid t\cdot e \mid p\cdot e \mid \muteq.c\\[0.8em]

c_\reset& ::= & \cmd{p_N}{e_\reset}\mid \cmd{p}{\reset} \\
e_\reset& ::= &  \mut a.c_\reset\tau\mid \mut[a_1.c_\reset\mid a_2.c'_\reset]   \\
&\mid & \mut(a_1,a_2).c_\reset \mid \mut(x,a).c_\reset\\
\end{array}
$\\
\multicolumn{2}{l}{
 $\begin{array}{l}  \textbf{\nef} \\\\
  \end{array}~
\begin{array}{c@{~}c@{~}l}
c_N 	& ::= & \cmd{p_N}{e_N}
\qquad\qquad
e_N 	 ::=  \star \mid \mut[a_1.c_N\mid a_2.c'_N] \mid \mut a.c_N\tau \mid \mut(a_1,a_2).c_N \mid \mut(x,a).c_N  \\
p_N,q_N & ::= & a \mid \injec{i}{p_N} \mid (p_N,q_N) \mid   (t,p_N) \mid\lambda x.p \mid \lambda a.p \mid \refl
	 \mid  \ind{t}{p_N}{ax}{q_N} \mid \cofix{t}{bx}{p_N}  \mid \mustar.c_N \mid \shift {c_\reset}\\
 
\end{array}$
}
\end{tabular}

}}
  \caption{The language of \dlpaw}
  \label{fig:language}
\end{figure*} 

The language of {\dlpaw} is based on the syntax of {\dltp}~\cite{Miquey17}, 
extended with the expressive power of {\dpaw}~\cite{Herbelin12} and with explicit stores as in the $\lbvtstar$-calculus~\cite{AriEtAl12}.
We stick to a stratified presentation of dependent types, that is to say 
that we syntactically distinguish terms—that represent \emph{mathematical objects}—from proof terms—that represent \emph{mathematical proofs}.
In particular, types and formulas are separated as well, matching the syntax of \dpaw's formulas.
Types are defined as finite types with the set of natural numbers as the sole ground type,
while formulas are inductively built on atomic equalities of terms,
by means of conjunctions, disjunctions, first-order quantifications, dependent products and co-inductive formulas:
$$\begin{array}{l@{\quad}r@{~~}c@{~~}l}
\text{\bf Types}        &T,U & ::= & \N \mid T\to U \\
\text{\bf Formulas    } &A,B  &::= &\top\mid \bot \mid t = u\mid A\land B \mid A\lor B \\
&&\mid&  \dptprod{a:A}{B}\mid\forall x^T. A \mid \exists x^T.A \mid \nu^t_{x,f}A
\end{array}$$

The syntax of terms is identical to the one in {\dpaw}, including functions $\lambda x.t$ 
and applications $t u$, as well as a recursion operator $\rec{t}{xy}{t_0}{t_S}$, 
so that terms represent objects in arithmetic of finite types.
As for proof terms (and contexts, commands), they are now defined with all the expressiveness of {\dpaw}.
Each constructor in the syntax of formulas is reflected by a constructor in the syntax of proofs 
and by the dual co-proof (\emph{i.e.} destructor) in the syntax of evaluation contexts.
Amongst other things, the syntax includes 
pairs $(t,p)$ where $t$ is a term and $p$ a proof, which inhabit the dependent sum type $\exists x^T\!.A$;
dual co-pairs $\tmu(x,a).c$ which bind the (term and proof) variables $x$ and $a$ in the command $c$;
functions $\lambda x.p$ inhabiting the type $\forall x^T\!.A$ together with their dual,
stacks $t\cdot e$ where $e$ is a context whose type might be dependent in $t$;
functions $\lambda a.p$ which inhabit the dependent product type $\dptprod{a:A}{B}$, 
and, dually, stacks $q\cdot e$, where $e$ is a context whose type might be dependent in $q$;
a proof term $\refl$ which is the proof of atomic equalities $t=t$ and
a destructor $\muteq.c$ which allows us to type the command $c$ modulo an equality of terms;
operators $\ind{t}{p_0}{ax}{p_S}$ and $\cofix{t}{bx}{p}$, as in \dpaw, for inductive and coinductive reasoning;
delimited continuations through proofs $\shift c_\tp$ and the context $\reset$;
a distinguished context $[]$ of type $\bot$, which allows us to reason ex-falso.

As in {\dltp}, the syntax of {\nef} proofs, contexts and commands is defined 
as a restriction of the previous syntax. Technically, they are defined (modulo $\alpha$-conversion)
with only one distinguished context variable $\star$ (and consequently only one binder $\mustar.c$), 
and without stacks of the shape $t\cdot e$ or $q\cdot e$ (to avoid applications).
Intuitively, one can understand {\nef} proofs as the proofs that cannot drop their continuation\footnote{See~\cite{Miquey17} for further details.}. 
The commands $c_\reset$ within delimited continuations are defined as commands of the shape $\cut{p}{\reset}$ or formed by a {\nef} proof
and a context of the shape $\mut a.c_\reset\tau$, $\mut[a_1.c_\reset| a_2.c'_\reset]$\vspace{-1mm},
$\mut(a_1,a_2).c_\reset$ or $\mut(x,a).c_\reset$.

We adopt a call-by-value evaluation strategy except for fixpoint operators\footnote{To highlight the duality 
between inductive and coinductive fixpoints, we evaluate both in a lazy way.
Even though this is not indispensable for inductive fixpoints, we find this approach more natural in that
we can treat both in a similar way in the small-step reduction system and thus through the realizability interpretation.}, 
which are evaluated in a lazy way. 
To this purpose, we use \emph{stores}\footnote{Our so-called \emph{stores} somewhat behave like lazy explicit substitutions or mutable environments. 
See~\cite{MiqHer18} for a discussion on this point.} in the spirit of the $\lbvtstar$-calculus, 
which are  defined as lists of bindings of the shape $[a:=p]$ where $p$ is 
a value or a (co-)fixpoint,
and of bindings of the shape $[\alpha:=e]$ where $e$ is any context.
We assume that each variable occurs at most once in a store $\tau$,
we thus reason up to $\alpha$-reduction and we assume the capability of generating fresh names.
Apart from evaluation contexts of the shape $\tmu a.c$ and co-variables $\alpha$,
all the contexts 
are \emph{forcing contexts} which eagerly require a value to be reduced and trigger the evaluation of lazily stored terms.
The resulting language is given in Figure\,\ref{fig:language}.

\begin{figure*}[t]
 \framebox{\vbox{\setlength{\saut}{0.5em}

\centering 
\begin{tabular}{c|c}
$\begin{array}{r@{~}c@{~}l}
  \multicolumn{3}{l}{\textrm{\bf Basic rules}} \\[0.3em]
  				\cmd{\lambda x.p}{V_t\cdot e}\tau & \red & \cmd{p[V_t/x]}{e}\tau 					\\
 (q\in\text{\nef}) \hfill   	\cmd{\lambda a.p}{q\cdot e}\tau & \red & \cmd{\shift\cmd{q}{\mut a.\cmd{p}{\reset}}}{e}\tau 	\\
 (q\notin\text{\nef})\qquad\hfill 	\cmd{\lambda a.p}{q\cdot e}\tau & \red & \cmd{q}{\mut a.\cmd{p}{e}}\tau 			\\
 (e\neq e_\reset)\hfill		\cmd{\mu\alpha.c}{e}\tau 	& \red & c\tau[\alpha:=e]                   			\\
				\cmd{V}{\mut a.c\tau'}\tau 	& \red & c\tau[a:=V]\tau'                    			\\[\saut]\hline
  \\[-1em]
  
  \multicolumn{3}{l}{\textrm{\bf Elimination rules}}  \\[0.3em]
  \cmd{\injec{i}{V}}{\mut[a_1.c_1\mid a_2.c_2]}\tau 	& \red& c_i\tau[a_i:=V] 			 \\
  \cmd{(V_1,V_2)}{\mut(a_1,a_2).c}\tau				& \red& c\tau[a_1:=V_1][a_2:=V_2]		 \\
  \cmd{(V_t,V)}{\mut(x,a).c}\tau 					& \red& (c[t/x])\tau[a:=V]			 \\
  \cmd{\refl}{\muteq.c}\tau		 				& \red& c\tau 				 \\[\saut]
  \\[-1em]
\end{array}$
\quad&\quad
  $\begin{array}{r@{~}c@{~}l}
     \multicolumn{3}{l}{\textrm{\bf Delimited continuations}} \\[0.3em]
(\text{if } c\tau \red c\tau')\hfill	  \cmd{\shift c}{e}\tau &\red& \cmd{\shift c}{e}\tau'     		\\
  \cut{\mu\alpha.c}{e_\reset}\tau					& \red& c[e_\reset/\alpha]\tau			\\
  \cut{\shift\cmd{p}{\reset}}{e}\tau				& \red& \cmd{p}{e}\tau				\\[\saut]\hline
  \\[-1em]
      
  \multicolumn{3}{l}{\textrm{\bf Call-by-value}} \\[0.3em]
  \text{($a$ fresh)} \hfill \cmd{\injec{i}{p}}{e}\tau 		 & \red& \cmd{p}{\mut a.\cut{\injec{i}{a}}{e}}\tau				\\
  \text{($a_1,a_2$ fresh)} \hfill \cmd{(p_1,p_2)}{e}\tau 	 & \red& \cmd{p_1}{\mut a_1.\cmd{p_2}{\mut a_2.\cut{(a_1,a_2)}{e}}}\tau 	\\
  \text{($a$ fresh)} \hfill \cmd{(V_t,p)}{e}\tau 		 & \red& \cmd{p}{\mut a.\cut{(V_t,a)}{e}}\tau 				\\[\saut]\hline
   \\[-1em]
  
  \multicolumn{3}{l}{\textrm{\bf Laziness}} \\
  \text{($a$ fresh)}\hfill \cmd{\cofix{V_t}{bx}{p}     }{e}\tau		& \red& \cut{a}{e}\tau[a:=\cofix{V_t}{bx}{p}] 			\\
  \text{($a$ fresh)}\hfill \cmd{\ind{V_t}{p_0}{bx}{p_S}}{e}\tau 	& \red& \cut{a}{e}\tau[a:=\ind{V_t}{p_0}{bx}{p_S}]		\\[\saut]
    \\[-1.5em]
  \end{array}$
  \\
  \multicolumn{2}{c}{}\\[-0.8em]\hline
  \multicolumn{2}{c}{
  $\begin{array}{r@{~}c@{~}l}
    \multicolumn{3}{l}{\textrm{\bf Lookup}} \\[0.3em]
  \cut{V}{\alpha}\tau[\alpha := e]\tau' 					& \red& \cut{V}{e}\tau[\alpha := e]\tau'	\\
  \cut{a}{f}\tau[a := V]\tau' 							& \red& \cut{V}{a}\tau[a:=V]\tau' 	 		\\[0.3em]
\qquad\qquad (b'\text{ fresh}) \qquad\qquad\qquad 	\cut{a}{f}\tau[a:=\cofix{V_t}{bx}{p}]\tau' 		& \red& \cmd{p[V_t/x][b'/b]}{\mut a.\cut{a}{f}\tau'}\tau[b':=\lambda y.\cofix{y}{bx}{p}] 	\\
				\cut{a}{f}\tau[a:=\ind{0}{p_0}{bx}{p_S}]\tau' 		& \red& \cmd{p_0}{\mut a.\cut{a}{f}\tau'}\tau 					\\
\qquad\qquad(b'\text{ fresh}) \hfill   	\cut{a}{f}\tau[a:=\ind{S(t)}{p_0}{bx}{p_S}]\tau' 		& \red& \cmd{p_S[t/x][b'/b]}{\mut a.\cut{a}{f}\tau'}\tau[b':=\ind{t}{p_0}{bx}{p_S}]	\\[\saut]\hline
  \\[-1em]
  \multicolumn{3}{l}{\textrm{\bf Terms }} \\[0.3em]
  \multicolumn{3}{c}{
  \begin{array}{c|c}
    \begin{array}{r@{~}c@{~}l}
     (\text{if } t\bred t')\hfill T[t]\tau 		& \red& T[t']\tau 			\\
   (\forall \alpha, \cut{p}{\alpha}\tau \red \cut{(t,p')}{\alpha}\tau)\hfill~~\qquad T[\wit p]\tau &\bred& T[t] \\
   (\lambda x.t)V_t & \bred& t[V_t/x]\\
 \rec{0}{xy}{t_0}{t_S}& \bred& t_0\\
\rec{S(u)}{xy}{t_0}{t_S}& \bred& t_S[u/x][\rec{u}{xy}{t_0}{t_S}/y]\\
  \end{array}
  &
  \begin{array}{l}

\multicolumn{1}{l}{\textrm{where:}} \\[0.3em]

\multicolumn{1}{l}{
\begin{array}{c@{~}c@{~}l}
 C_t\hole 	& ::= & \cmd{(\hole,p)}{e} \mid \cmd{\ind{\hole}{p_0}{ax}{p_S}}{e}\\
 &\mid& \cmd{\cofix{\hole}{bx}p}{e}\mid\cut{\lambda x.p}{\hole\cdot e}\\[0.5em]
 T\hole 	& ::= & C_t\hole \mid T[\hole u] \mid T[\rec{\hole}{xy}{t_0}{t_S}] \\
  \end{array}
  }
\end{array}
  \end{array}
}\\
  \end{array}$
}
\end{tabular}

}}
 \caption{Reduction rules of \dlpaw}
 \label{fig:bigstep}
\end{figure*}
\begin{figure*}[t]
\framebox{\vbox{\input{types}}}
  \caption{Type system for \dlpaw}
 \label{fig:dlpaw_types}
\end{figure*}
\subsection{Reduction rules}
The reduction system of {\dlpaw} is given in Figure\,\ref{fig:bigstep}.
The basic rules are those of the call-by-value \lmmt-calculus and of {\dltp}. 
The rules for delimited continuations are exactly the same as in {\dltp},
except that we have to prevent $\reset$ from being caught and stored by a proof $\mu\alpha.c$.
We thus distinguish two rules for commands of the shape $\cut{\mu\alpha.c}{e}$, depending
on whether $e$ is of the shape $e_\reset$ or not. 
In the former case, we perform the substitution $[e_\reset/\alpha]$, 
which is linear since $\mu\alpha.c$ is necessarily {\nef}.
We should also mention in passing that we abuse the syntax in every other rules, 
since $e$ should actually refer to $e$ or $e_\tp$ 
(or the reduction of delimited continuations would be stuck).
Elimination rules correspond to commands where the proof is a constructor (say of pairs) applied to values,
and where the context is the matching destructor. 
Call-by-value rules correspond to ($\varsigma$) rule of Wadler's sequent calculus~\cite{Wadler03}.
The next rules express the fact that (co-)fixpoints are lazily stored, and reduced only if their value is eagerly demanded 
by a forcing context. 
Lastly, terms are reduced according to the usual $\beta$-reduction, 
with the operator $\texttt{rec}$ computing with the usual recursion rules. 
It is worth noting that the stratified presentation allows to define the reduction of terms as external: 
within proofs and contexts, terms are reduced in place. Consequently, as in {\dltp} the very same happen for 
{\nef} proofs embedded within terms. 
Computationally speaking, this corresponds indeed to the intuition that terms are reduced on an external device.


\subsection{Typing rules}
As often in Martin-L\"of's intensional type theory,
formulas are considered up to equational theory on terms.
We denote by $A\equiv B$ the reflexive-transitive-symmetric closure of the relation $\typered$ induced by the reduction of terms and {\nef} 
proofs as follows:
\begin{center}
\begin{tabular}{lcl}
   $A[t]$ &$\typered$ &$A[t']\quad$ whenever $\quad t\rightarrow_\beta t'$ \\
   $A[p]$ &$\typered$& $A[q] \quad$ whenever $\quad \forall \alpha\,(\cmd{p}{\alpha}\rightarrow\cmd{q}{\alpha})$ \\
\end{tabular}
\end{center}
in addition to the reduction rules for equality and for coinductive formulas:
$$
\begin{array}{r@{~~}c@{~~}l}
    0 = S(t) & \typered & \bot \\
    S(t) = 0 & \typered & \bot \\
\end{array}
\qquad
\begin{array}{r@{~~}c@{~~}l}
    S(t) = S(u) & \typered & t = u \\
   \nu ^t_{fx} A &\typered& A[t/x][\nu^y_{fx} A/f(y)=0] 
\end{array}$$

We work with one-sided sequents where typing contexts are defined by:
$$
\begin{array}{rcl}\Gamma,\Gamma' & ::= & \eps \mid \Gamma,x:T \mid \Gamma,a:A \mid \Gamma,\alpha:A^\negt\mid \Gamma,\reset:A^\negt.
\end{array}
$$
using the notation $\alpha:A^\negt$ for an 
assumption of the refutation of $A$.
This allows us to mix hypotheses over terms, proofs and contexts while keeping 
track of the order in which they are added (which is necessary because of the dependencies).
We assume that a variable occurs at most once in a typing context.

We define nine syntactic kinds of typing judgments:
six\footnote{For terms, proofs, contexts, commands, closures and stores.} in regular mode, that we write $\Gamma \sigdash J$,
and three\footnote{For contexts, commands and closures.} more for the dependent mode, that we
write $\Gamma\vdash_d J;\sigma$.
In each case, $\sigma$ is a list of dependencies---we explain the presence of a list of dependencies in each case thereafter---, which are defined from the following grammar:
$$
\sigma ::= \varepsilon \mid \sigma\dpt{p|q}
$$
The substitution on formulas according to a list of dependencies $\sigma$ is defined by:\nomidem
$$
\eps(A) \defeq \{A\}\qquad\qquad\qquad
\sigma\dpt{p|q}(A) \defeq \begin{cases}
 \sigma (A[q/p]) & \text{if $q\in\nef$} \\
 \sigma(A) & \text{otherwise}
\end{cases}
$$
Because the language of proof terms include constructors for pairs, injections, etc, 
the notation $A[q/p]$ does not refer to usual substitutions properly speaking: $p$ can be a pattern (for instance $(a_1,a_2)$)
and not only a variable.

We shall attract the reader's attention to the fact that 
all typing judgments include a list of dependencies.
Indeed, as in the $\lbvtstar$-calculus, when a proof or a context is caught by a binder, say $V$ and $\tmu a$, 
the substitution $[V/a]$ is not performed but rather put in the store: $\tau[a:=V]$.
Now, consider for instance the reduction of a dependent function $\lambda a.p$ (of type $\dptprod{a:A}{B}$)
applied to a stack $V\cdot e$\footnote{We refer the reader to \cite{Miquey17} for detailed explanations on this rule.}:
\begin{align*}\cut{\lambda a .p}{V\cdot e}\tau  &\red \cut{\shift\cut{V}{\tmu a.\cut{p}{\reset}}}{e}\tau\\
&\red \cut{\shift\cut{p}{\reset}}{e}\tau[a:=V]
 \red \cut{p}{e}\tau[a:=V]\nomidem
\end{align*}
Since $p$ still contains the variable $a$, whence his type is still $B[a]$, whereas the type of $e$ is $B[V]$.
We thus need to compensate the missing substitution\footnote{On the contrary, 
the reduced command in {\dltp}  would have been $\cut{p[V/a]}{e}$, which is 
typable with the \cutrule~rule over the formula $B[V/a]$.}. 

We are mostly left with two choices.
Either we mimic the substitution in the type system, which would amount to the following typing rule:
$$\qquad
\infer{\sigmaopt\Gamma\vdash c\tau}{\Gamma,\Gamma'\vdash \tau(c) & \Gamma \vdash \tau:\Gamma'}
$$
$$\hspace{-3mm}\raisebox{0.6em}{\mbox{$
\begin{array}{l@{~~}|@{~~}l}
\begin{array}{l}
 \text{where:}\\
 \tau[\alpha:=e](c) \defeq \tau(c) \\
\end{array}
&\begin{array}{l}
 \tau[a:=p_N](c) \defeq \tau(c[p_N/a]) \quad~~ (p\in\nef)\\
 \tau[a:=p](c) \defeq \tau(c)\hfill(p\notin\nef) \\
\end{array}
\end{array}
$}
}
$$
Or we type stores in the spirit of the $\lbvtstar$-calculus, 
and we carry along the derivations all the bindings liable to be used in types, which constitutes again a list of dependencies.

The former solution has the advantage of solving the problem before typing the command, 
but it has the flaw of performing computations which would not occur in the reduction system.
For instance, the substitution $\tau(c)$ could duplicate co-fixpoints (and their typing derivations), which would never happen in the calculus.
That is the reason why we favor the other solution, which is closer to the calculus in our opinion. 
Yet, it has the drawback that it forces us to carry a list of dependencies even in regular mode. 
Since this list is fixed (it does not evolve in the derivation except when stores occur), 
we differentiate the denotation of regular typing judgments, written $\Gamma \sigdash J$, from the one of judgments in dependent mode, 
which we write $\Gamma \vdash_d J;\sigma$ to highlight that $\sigma$ grows along derivations.
The type system we obtain is given in \Cref{fig:dlpaw_types}.

\subsection{Subject reduction}
\label{s:sub_red}



We shall now prove that typing is preserved along reduction. As for the $\lbvtstar$-calculus,
the proof is simplified by the fact that substitutions are not performed (except for terms), which keeps us from proving 
the safety of the corresponding substitutions.
Yet, we first need to prove some technical lemmas about dependencies. 
To this aim, we define a relation $\sigma\dptimp \sigma'$
between lists of dependencies, which expresses the fact that any typing derivation obtained with $\sigma$ could be obtained as well as with $\sigma'$:
 $$\sigma \dptimp \sigma' ~\defeq~\sigma(A) =\sigma(B) \limp \sigma'(A) =\sigma'(B) \eqno(\text{for any}~ A,B)$$

\begin{proposition}[Dependencies weakening]\label{prop:dlpaw:dpt_weak} 
 If $\sigma,\sigma'$ are two lists of dependencies such that $\sigma\dptimp\sigma'$, then any derivation using $\sigma$ can be done using $\sigma'$ instead. 
 In other words, the following rules are admissible:
 $$
 \infer[\weakrule]{\Gamma \vdash^{\sigma'} J}{\Gamma \vdash^{\sigma} J}
 \qquad\qquad\qquad
 \infer[\weakdrule]{\Gamma \vdash_d J;\sigma'}{\Gamma \vdash_d J;\sigma}
 $$
\end{proposition}

We can prove the safety of reduction with respect to typing:
\begin{theorem}[Subject reduction]
\label{thm:subject_reduction}
For any context $\Gamma$ and any closures $c\tau$ and $c'\tau'$ such that $c\tau \rightarrow c'\tau'$, we have:
\begin{center}
 1.~If ~$\Gamma\vdash c\tau$~then~$\Gamma\vdash c'\tau'$.\qquad
 2.~If ~$\Gamma\vdash_d c\tau;\varepsilon$~then~$\Gamma\vdash_d c'\tau';\varepsilon$.
 \end{center}
\end{theorem}
\begin{proof}
 The proof follows the usual proof of subject reduction, by induction on the reduction $c\tau \rightarrow c'\tau'$.
\vlong{See Appendix~\ref{app:sub_red}.}
\end{proof}

\begin{figure*}[t]
 \frame{\vbox{
 \input{ded_nat}
 }}\vspace{-1em}
 \caption{Typing rules of $\dpaw$}
 \label{dpaw_types}
\end{figure*}

\subsection{Natural deduction as macros}
\label{s:macros}
We can recover the usual proof terms for elimination rules in natural deduction systems, 
and in particular the ones from $\dpaw$,
by defining them as macros in our language. 
The definitions are straightforward, using delimited continuations
for $\letop \dots\inop $ and the constructors over {\nef} proofs which
might be dependently typed:
$$
\begin{array}{l}
\begin{array}{@{}r@{~~\defeq~~}l}
  {\letin{a}{p}q} 		& \mu\alpha_p.\cmd{{p}}{\mut a.\cmd{{q}}{\alpha_p}} \\
  {\split{p}{a_1,a_2}q} 	& \mu\alpha_p.\cmd{{p}}{\mut (a_1,a_2).\cmd{{q}}{\alpha_p}} \\
  {\case{p}{a_1.p_1}{a_2.p_2}}  & \mu\alpha_p.\cmd{{p}}{\mut[a_1.\cmd{{p_1}}{\alpha_p}|a_2.\cmd{{p_2}}{\alpha_p}]}\\
  {\dest{p}{a,x}q}  		& \mu\alpha_p.\cmd{{p}}{\mut (x,a).\cmd{{q}}{\alpha_p}}  \\
  \prf p 			&  \shift\cmd{p}{\tmu (x,a).\cmd{a}{\reset}}\\
\end{array}\\
\\[-0.9em]
\begin{array}{c@{~~}|@{~~}c}
\begin{array}{r@{~\defeq~}l@{}}
  \subst p q			& \mu \alpha.\cmd{p}{\muteq.\cmd{q}{\alpha}} \\
  {\exf{p}}  			& \mu\alpha.\cmd{{p}}{[]} 
\end{array}&
\begin{array}{r@{~\defeq~}l@{}}
  {\catch{\alpha}p}  		& \mu\alpha.\cmd{{p}}{\alpha} \\
  {\throw{\alpha}~p}  		& \mu\_.\cmd{{p}}{\alpha}\\
\end{array}
\end{array}
\end{array}
$$
where $\alpha_p = \reset$ if $p$ is {\nef} and $\alpha_p = \alpha$ otherwise.

It is then straightforward to check that the macros match the expected typing rules:
\begin{proposition}[Natural deduction]
 The typing rules from $\dpaw$, given in Figure~\ref{dpaw_types}, are admissible.
\end{proposition}

One can even check that the reduction rules in {\dlpaw} for these proofs almost mimic the ones of \dpaw. 
To be more precise, the rules of {\dlpaw} do not allow to simulate each rule of {\dpaw}, 
due to the head-reduction strategy amongst other things. 
Nonetheless, up to a few details the reduction of a command in {\dlpaw} follows one particular reduction path of 
the corresponding proof in {\dpaw}, or in other words, one reduction strategy.

\newcommand{\nthn}[2]{\nth_{#1}\,#2}
The main result is that using the macros, the same proof terms are suitable for
countable and dependent choice~\cite{Herbelin12}.
We do not state it here, but following the approach of~\cite{Herbelin12}, 
we could also extend {\dlpaw} to obtain a proof for the axiom of bar induction.
\begin{theorem}[Countable choice~\cite{Herbelin12}]\label{thm:acn}
 We have:
$$\begin{array}{@{~~}rc@{~~}l@{}l}
AC_\N		&:=& \lambda H.&\letop a=\cofix{0}{bn}{(H n, b(S(n))}\\
&&&\,\inop\, (\lambda n.\wit(\nthn n a),\lambda n.\prf(\nthn n a)\\
		&: & \multicolumn{2}{l}{\forall x^\N \exists y^T P(x,y) \imp \exists f^{\N\imp T} \forall x^\N P(x,f(x))}\\
\end{array}$$
where $\nthn n a:= \pi_1(\ind{n}{a}{x,c}{\pi_2(c)})$.
\end{theorem}


\begin{theorem}[Dependent choice~\cite{Herbelin12}]\label{thm:dc}
 We have:
$$\begin{array}{r@{~~}c@{~~}l@{}l}
DC &:=& \lambda H.\lambda x_0.\letop~ a=(x_0,\cofix{0}{bn}{d_n})fsix~\\ 
   & & \inop~ (\lambda n.\wit(\nthn n a),(\refl,\lambda n.\pi_1(\prf(\prf(\nthn n a)))))\\
   &:& \forall x^T\!. \exists y^T\!. P(x,y) \imp \\ 
&& \qquad  \forall x_0^T\! .\exists  f\in T^\N.( f(0) = x_0 
 \land \forall n^\N\!. P( f (n), f (s (n))))\\
\end{array}$$
where $d_n:=\dest{H n}{y,c} (y,(c,b\,y)))$\\
and ~~~
$\nthn n a:= \ind{n}{a}{x,d}{(\wit(\prf d),\pi_2(\prf(\prf(d))))}$.
\end{theorem}


\section{Small-step calculus}
\label{s:small_step}
As for the $\lbvtstar$-calculus~\cite{AriEtAl12,MiqHer18}, we follow here Danvy's methodology of semantic artifacts~\cite{DanEtAl10,AriEtAl12}
to obtain a realizability interpretation.
We first decompose the reduction system of {\dlpaw} into small-step reduction rules, that we denote by $\reds$. 
This requires a refinement and an extension of the syntax, that we shall now present.
To keep us from boring the reader stiff with new (huge) tables for the syntax, typing rules and so forth,
we will introduce them step by step. 
We hope it will help the reader to convince herself of the necessity and of the somewhat naturality of these extensions.

\mypar{Values}
First of all, we need to refine the syntax to distinguish between strong and weak values
in the syntax of proof terms.
As in the $\lbvtstar$-calculus, this refinement is induced by the computational behavior of the calculus:
weak values are the ones which are stored by $\tmu$ binders, but which are not
values enough to be eliminated in front of a forcing context, that is to say variables.
Indeed, if we observe the reduction system, we see that in front of a forcing context $f$, 
a variable leads a search through the store for a ``stronger'' value, which could incidentally provoke the
evaluation of some fixpoints. 
On the other hand, strong values are the ones which can be
reduced in front of the matching forcing context, that is to say functions, $\refl$, pairs of values, 
injections or dependent pairs:
$$
\begin{array}{>{\!\!}lc@{~}c@{~}l}
\textbf{Weak values}   & V	& ::= & a  \mid v \\
\textbf{Strong values} & v	& ::= & \injec{i}{V} \mid (V,V) \mid (V_t,V)  \mid  \lambda x.p \mid \lambda a.p\mid \refl 
\end{array}
$$
This allows us to distinguish commands of the shape $\cut{v}{f}\tau$, where the forcing context (and next the strong
value) are examined to determine whether the command reduces or not; from commands of the shape $\cut{a}{f}\tau$
where the focus is put on the variable $a$, 
which leads to a lookup for the associated proof in the store.

\mypar{Terms}
Next, we need to explicit the reduction of terms. To this purpose,
we include a machinery to evaluate terms in a way which resemble the evaluation of proofs. 
In particular, we define new commands which we write $\cut{t}{\pi}$ where 
$t$ is a term and $\pi$ is a context for terms (or co-term).
Co-terms are either of the shape $\tmu x.c$ or stacks of the shape $u\cdot \pi$. 
These constructions are the usual ones of the \lmmt-calculus (which are also the ones for proofs).
We also extend the definitions of commands with delimited continuations to include the corresponding commands for terms:
$$
\begin{array}{c@{~}|@{~}c}
\begin{array}{@{}l}
  \textbf{Commands}\\
  \textbf{Co-terms} \\
\end{array}\quad
\begin{array}{c@{~}c@{~}l}
c	& ::= & \cmd{p}{e} \mid \cut{t}{\pi}\\
\pi 	& ::= & t\cdot \pi \mid \mut x.c \\ 
\end{array}
&
\begin{array}{c@{~}c@{~}l}
c_\reset& ::= & \cdots\mid \cmd{t}{\pi_\reset}\\
\pi_\reset & ::= & t\cdot \pi_\reset \mid \mut x.c_\reset
\end{array}
\end{array}
$$
We give typing rules for these new constructions, which are the usual rules for typing contexts in the \lmmt-calculus:
$$
\infer[\implrule]{\Gamma \vdash t \cdot \pi:(T\to U)^\negt}
      {\Gamma \vdash t:T  & \Gamma\vdash \pi:U^\negt}
\qquad	\qquad
\infer[\autorule{\tmu_x}]{\Gamma\vdash \tmu x.c : T^\negt }
      {c : (\Gamma, x:T)}
$$
$$
\infer[\autorule{\textsc{cut}_t}]{\Gamma\sigdash \cmd{t}{\pi}}{\Gamma\sigdash {t}:{T}  & \Gamma\sigdash \pi :T^\negt }
$$
It is worth noting that the syntax as well as the typing and reduction rules for terms now match exactly
the ones for proofs\footnote{Except for substitutions of terms, which we could store as well.}. In other words, with these definitions, 
we could abandon the stratified presentation without any trouble, since reduction rules for
terms will naturally collapse to the ones for proofs.

\mypar{Co-delimited continuations}
Finally, in order to maintain typability when reducing dependent pairs of the strong existential type,
 we need to add what we call \emph{co-delimited continuations}. 
As observed in~\cite{Miquey17}, the \cps~translation of pairs $(t,p)$ in {\dltp} is not the expected one, reflecting the need for a special reduction rule. 
Indeed, consider such a pair of type $\exists x^T\!. A$, the standard way of reducing it would be a rule like:
$$\cut{(t,p)}{e}\tau \reds \cut{t}{\tmu x.\cut{p}{\tmu a.\cut{(x,a)}{e}}}\tau$$
but such a rule does not satisfy subject reduction.
Consider indeed a typing derivation for the left-hand side command, when typing the pair $(t,p)$, $p$ is of type $A[t]$. 
On the command on the right-hand side, the variable $a$ will then also be of type $A[t]$, while it should be of type $A[x]$ for the pair $(x,a)$ to be typed.
We thus need to compensate this mismatching of types, by reducing $t$ within a context where $a$ is not linked to $p$ but to a co-reset $\coreset$ (dually to reset $\reset$), 
whose type can be changed from $A[x]$ to $A[t]$ thanks to a list of dependencies:
$$\cmdp{(t,p)}{e}\tau    \reds  \cmdp{p}{\coshift \cmd{t}{\mut x.\cmd{\coreset}{\tmu a.\cut{(x,a)}{e}}}}\tau 	$$
We thus equip the language with new contexts $\coshift c_\coreset$, which we call \emph{co-shifts} and where $c_\coreset$ is a command
whose last cut is of the shape $\cut{\coreset}{e}$.
This corresponds formally to the following syntactic sets, which are dual to the ones introduced for delimited continuations:
$$
\begin{array}{ccl}
e	& ::= & \cdots \mid  \coshift {c_\coreset} \\[0.8em]  
c_\coreset& ::= & \cmd{p_N}{e_\coreset}\mid \cmd{t}{\pi_\coreset}\mid \cut{\coreset}{e}\\
e_\coreset& ::= &  \mut a.c_\coreset\mid \mut[a_1.c_\coreset\mid a_2.c'_\coreset]  \\ & \mid &  \mut(a_1,a_2).c_\coreset \mid \mut(x,a).c_\coreset \\
\pi_\coreset & ::= & t\cdot \pi_\coreset \mid \mut x.c_\coreset \\[0.8em]   
e_N 	& ::= & \cdots \mid \coshift c_\coreset  \\
\end{array}\leqno 
\begin{array}{@{}l}
  \textbf{Contexts}\\[0.8em]
  \textbf{Co-delimited}\\ 
  \textbf{continuations}\\\\\\[0.8em] 
  \textbf{\nef}\\ 
  \end{array}$$
This might seem to be a heavy addition to the language, 
but we insist on the fact that these artifacts are merely 
the dual constructions of delimited continuations
introduced in \dltp, with a very similar intuition. 
In particular, it might be helpful for the reader to think of the fact that
we introduced delimited continuations for type safety 
of the evaluation of dependent products in $\dptprod{a:A}{B}$ 
(which naturally extends to the case $\forall x^T\!.A$).
Therefore, to maintain type safety of dependent sums in $\exists x^T\!.A$,
we need to introduce the dual constructions of co-delimited continuations.
We also give typing rules to these constructions, which are dual to the typing rules for delimited-continuations:
$$
\scalebox{0.97}{\infer[\!\!\autorule{\tmu\coreset}]{\Gamma\sigdash \coshift c_\coreset : A^\negt}{
\Gamma,\coreset:A\vdash_d c_\coreset;\sigma
  }
\qquad\quad
\infer[\!\!\autorule{\coreset}]{\Gamma,\coreset:B,\Gamma' \vdash_d \cmd{\coreset}{e};\sigma}{\Gamma,\Gamma'\sigdash e:A^\negt & \sigma(A)=\sigma(B)}
}
$$
Note that we also need to extend the definition of list of dependencies to include bindings of the shape $\dpt{x|t}$ for terms, 
and that we have to give the corresponding typing rules to type commands of terms in dependent mode:
$$
\scalebox{0.97}{
\infer[\!\!\dxmutrule]{\Gamma\vdash_d \tmu x.c : T^\negt;\sigma\rdpt{t}\! }{c : (\Gamma, x:T;\sigma\dpt{x|t})}
~~~
\infer[\!\!\dtcutrule]{\Gamma,\coreset:B,\Gamma' \vdash_d \cmd{t}{\pi};\sigma}{\Pi_t\!\! & \Gamma,\coreset:B,\Gamma'\vdash_d \pi:A^\negt;\sigma\dpt{\cdot|t}\! }
}
$$
where $\Pi_t \defeq \Gamma,\Gamma'\sigdash {t}:{T}$.
  
\vlong{The small-step reduction system is given in Appendix~\ref{app:small_step}.
The }
\vshort{Finally, small-step reduction}
rules are written $c_\iota\tau \reds c'_o\tau'$
where the annotation $\iota,p$ on commands 
are indices (\emph{i.e.} $c,p,e,V,f,t,\pi,V_t$) indicating which
part of the command is in control. 
As in the $\lbvtstar$-calculus, we observe an alternation of steps 
descending from $p$ to $f$ for proofs and from $t$ to $V_t$ for terms.
The descent for proofs can be divided in two main phases.
During the first phase, from $p$ to $e$ we observe the call-by-value 
process, which extracts values from proofs, opening recursively the 
constructors and computing values. 
In the second phase, the core computation takes place from $V$ to $f$, 
with the destruction of constructors and the application of function to
their arguments.
The laziness corresponds precisely to a skip of the first phase, waiting
to possibly reach the second phase before actually going through the first one.
  

Here again, reduction is safe with respect to the type system:
\begin{proposition}[Subject reduction]\label{p:small_sub_red}
 The small-step reduction rules satisfy subject reduction.
\end{proposition}
\vlong{\begin{proof}
 The proof is again an induction on $\reds$, see Appendix~\ref{app:small_step}.
\end{proof}}

It is also direct to check that the small-step reduction system simulates the big-step one, 
and in particular that it preserves the normalization :
\begin{proposition}\label{prop:reduction}
 If a closure $c\tau$ normalizes for the reduction $\reds$, then it normalizes for $\red$.
\end{proposition}
\vlong{
\begin{proof}
 By contraposition, see Appendix~\ref{app:small_step}.
\end{proof}
}

\section{A realizability interpretation of {\textbf{\boldmath dLPA$^\omega$}}}
\label{s:realizability}
We shall now present the realizability interpretation of $\dlpaw$, which will finally give us a proof of its normalization.
Here again, the interpretation combines ideas of the interpretations for the $\lbvtstar$-calculus~\cite{MiqHer18}
and for {\dltp} through the embedding in Lepigre's calculus \cite{Miquey17,Lepigre16}.
Namely, as for the $\lbvtstar$-calculus, formulas will be interpreted by sets of proofs-in-store 
of the shape $\tis{p}{\tau}$, and the orthogonality will be defined between proofs-in-store $\tis{p}{\tau}$ and contexts-in-store $\tis{e}{\tau'}$ such
that the stores $\tau$ and $\tau'$ are compatible.

We recall the main definitions necessary to the realizability interpretation:
\begin{definition}[Proofs-in-store]
We call \emph{closed \cp} (resp. \emph{closed \ce}, \emph{closed \ct}, etc) 
the combination of a proof $p$ (resp. context $e$, term $t$, etc) with a closed store $\tau$ such that
$FV(p)\subseteq \dom(\tau)$. 
We use the notation $\tis{p}{\tau}$ to denote such a pair. 
In addition, we denote by $\Lambda_p$ (resp. $\Lambda_e$, etc.) the set of all proofs 
and by $\Lambda_p^\tau$ (resp. $\Lambda_e^\tau$, etc.) the set of all proofs-in-store.

We denote the sets of closed closures by $\C_0$, and we identify $\tis{c}{\tau}$ with the closure $c\tau$ when $c$ is closed in $\tau$.
\end{definition}

We now recall the notion of compatible stores~\cite{MiqHer18}, which allows us to define
an orthogonality relation between proofs- and contexts-in-store.
\begin{definition}[Compatible stores and union]  
Let $\tau$ and $\tau'$ be stores, we say that:
\begin{itemize}
 \item they are \emph{independent} and note $\indpt{\tau}{\tau'}$ if
 ${\dom(\tau)\cap\dom(\tau')=\emptyset}$.
 \item they are \emph{compatible} and note $\compat{\tau}{\tau'}$ 
 if for all variables $a$ (resp. co-variables $\alpha$) present in both stores: ${a\in \dom(\tau)\cap\dom(\tau')}$;
 the corresponding proofs (resp. contexts) in $\tau$ and $\tau'$ 
 coincide. 
 \item $\tau'$ is an \emph{extension} of $\tau$ and we write $\tau\stext \tau'$ whenever 
 $\compat{\tau}{\tau'}$ and $\dom(\tau)\subseteq\dom(\tau')$.

 \item $\overline{\tau\tau'}$ is \emph{the compatible union } of compatible closed stores $\tau$ and $\tau'$.
 It is defined as $\overline{\tau\tau'}\defeq \stjoin{\tau}{\tau'}$, which itself given by:\vspace{-0.4em}
 $$\begin{array}{r@{~~\defeq~~}l}
  \stjoin{\tau_0[a:=p]\tau_1}{\tau'_0[a:=p]\tau'_1} & \tau_0\tau'_0[a:=p]\stjoin{\tau_1}{\tau'_1} \\
  \stjoin{\tau_0[\alpha:=e]\tau_1}{\tau'_0[\alpha:=e]\tau'_1} & \tau_0\tau'_0[\alpha:=e]\stjoin{\tau_1}{\tau'_1} \\
  \stjoin{\tau_0}{\tau_0'} & \tau_0\tau_0'                                                             \\
 \end{array}\nomidem
 $$
 where $\indpt{\tau_0}{\tau_0'}$.
\end{itemize}
 \end{definition}
The next lemma (which follows from the previous definition) states the main property  
we will use about union of compatible stores.
\begin{lemma}
\label{lm:st_union}
If $\tau$ and $\tau'$ are two compatible stores, then $\tau\stext\overline{\tau\tau'}$ and $\tau'\stext\overline{\tau\tau'}$.
Besides, if $\tau$ is of the form $\tau_0[x:=t]\tau_1$,
then $\overline{\tau\tau'}$ is of the form $\overline{\tau_0}[x:=t]\overline{\tau_1}$ with $\tau_0 \stext \overline{\tau_0}$
and $\tau_1\stext\overline{\tau_1}$.
\end{lemma}

We can now define the notion of pole, which has to satisfy an extra condition due to the presence of delimited continuations
\begin{definition}[Pole]
 A subset $\pole\in \C_0$ is said to be \emph{saturated} or \emph{closed by anti-reduction} 
 whenever for all $\tis{c}{\tau},\tis{c'}{\tau'}\in\C_0$, we have:
 $$(c'\tau' \in \pole) ~ \land ~(c\tau\rightarrow c'\tau')~\Rightarrow~ (c\tau\in\pole)$$
 It is said to be \emph{closed by store extension} if whenever $c\tau$ is in $\pole$, for any store $\tau'$ extending $\tau$,
 $c\tau'$ is also in $\pole$:
 $$(c\tau\in\pole) ~\land~ (\tau\stext\tau') ~\Rightarrow~ (c\tau'\in\pole)$$
 It is said to be \emph{closed under delimited continuations} if whenever $c[e/\reset]\tau$ (resp. $c[V/\coreset]\tau$) is in $\pole$, 
 then $\cut{\shift c}{e}\tau$ (resp .$\cut{V}{\coshift c}\tau$) belongs to $\pole$:\nomidem
 $$(c[e/\reset]\tau\in\pole) ~\Rightarrow~ (\cut{\shift c}{e}\tau\in\pole)\nomidem$$
 $$ (c[V/\coreset]\tau\in\pole) ~\Rightarrow ~(\cut{V}{\coshift c}\tau\in\pole)$$
 A \emph{pole} is defined as any subset of $\C_0$ that is closed by anti-reduction, by store extension and under delimited continuations.
\end{definition}

We verify that the set of normalizing command is indeed a pole:
\begin{proposition}
\label{prop:dlpaw:norm_pole}
 The set $\pole_{\Downarrow}=\{c\tau\in\C_0:~c\tau\text{ normalizes }\}$ is a pole.
\end{proposition}
 \begin{figure*}[t]
  \myfig{
$
\begin{array}{>{\hspace{-0.2cm}}r@{~}c@{~}l@{~}|@{~}l}
     \fvf{\bot}	    & \defeq & \Lambda^\tau_f\\ 
     \fvf{\top}	    & \defeq & \emptyset\\ 
     \fvf{\dot F (t)}& \defeq & F(t)\\
 \fvf{\exists x.A}  & \defeq & \bigcap_{t\in\Lambda_t} \fvf{A[t/x]}\\
 \fvf{\forall x.A}  & \defeq & (\bigcap_{t\in\Lambda_t} \fvf{A[t/x]}^{\pole_v})^{\pole_f}\\
 \fvf{\forall a.A}  & \defeq & (\bigcap_{t\in\Lambda_p} \fvf{A[p/a]}^{\pole_v})^{\pole_f}\\
 \fvf{\nu^t_{fx}A}  & \defeq & \bigcup_{n\in\N}\fvf{F^n_{A,t}}\\ 
     \tvV{A} 	    & \defeq & \fvf{A}^{\pole_V} \\
     \fve{A} 	    & \defeq & \tvV{A}^{\pole_e} \\

     \multicolumn{4}{l}{}\\
     \tvVt{\N} 	 & \defeq & \{\tis{S^n(0)}{\tau}, n\in\N\}\\
 \tvVt{t\in T}   & \defeq & \{\tis{V_t}{\tau} \in \tvVt{T}: V_t \eqtau t \}\\
 \tvVt{T\to U} 	 & \defeq & \multicolumn{2}{l}{\{\tis{\lambda x .t}{\tau}  : \forall V_t \tau', \compat{\tau}{\tau'}\land \tis{V_t}{\tau'}\in\tvVt{T} \Rightarrow \tis{t[V_t/x]}{\overline{\tau\tau'}}\in\tvt{U}\}}\\
   \end{array}
   \hspace{-5.4cm}
   \begin{array}{r@{~}c@{~}l}
     \fvf{t=u}	    & \defeq & \begin{cases}\{\tis{\muteq.c}{\tau}: c\tau\in\pole\}& \text{if~} t\equiv_\tau u\\ \Lambda^\tau_f & \text{otherwise}\end{cases}\\
     \fvf{p\in A}   & \defeq & \{\tis{V}{\tau}\in\tvV{A}: V \eqtau p \}^{\pole_f} \\
     \fvf{T\imp B}  & \defeq & \{\tis{V_t\cdot e}{\tau} : \tis{V_t}{\tau}\in\tvVt{t\in T} \land \tis{e}{\tau}\in \fve{B}\}\\
     \fvf{A\imp B}  & \defeq & \{\tis{V\cdot e}{\tau}   : \tis{V}{\tau}\in\tvV{A} \land \tis{e}{\tau}\in \fve{B}\}\\
     \fvf{T\land A} & \defeq & \{\tis{\tmu(x,a).c}{\tau}: \forall \tau',V_t\in\tvVt{T}^{\tau'},{V}\in\tvV{A}^{\tau'},\compat{\tau}{\tau'}\Rightarrow c[V_t/x]\overline{\tau\tau'}[a:=V]\in\pole\}\\
 \fvf{A_1\land A_2} & \defeq & \{\tis{\tmu(a_1,a_2).c}{\tau}: \forall \tau',V_1\in \!\tvV{A_1}^{\tau'},V_2\in\!\tvV{A_2}^{\tau'},\compat{\tau}{\tau'}\Rightarrow c\overline{\tau\tau'}[a_1:=V_1][a_2:=V_2]\!\in\!\pole\}\\
 \fvf{A_1\lor A_2}  & \defeq & \{\tis{\tmu[a_1.c_1|a_2.c_2]}{\tau}: \forall\tau',{V}\in\tvV{A_i}^{\tau'},\compat{\tau}{\tau'}\Rightarrow c\overline{\tau\tau'}[a_i:=V]\in\pole\}\\
\tvp{A} 	    & \defeq & \fve{A}^{\pole_p} \\
     \\     
     \fvpi{T} 	 & \defeq & \tvVt{A}^{\pole_\pi}\\
     \tvt{T} 	 & \defeq & \fvpi{A}^{\pole_t}\,\\
              \\              
   \end{array}
   $
\vspace{-6em}   
\begin{flushright}\begin{tabular}{|l}
\begin{minipage}{0.35\textwidth}
where:
\begin{itemize}
 \item $p \in S^{\tau}$ (resp. $e$,$V$,etc.) denotes $\tis{p}{\tau} \in S$ (resp. $\tis{e}{\tau}$, $\tis{V}{\tau}$, etc.),
\item $F$ is a function from $\Lambda_t$ to $\P(\Lambda_f^\tau)_{/{\equiv_\tau}}$.
\end{itemize}
\end{minipage}
\end{tabular}
\end{flushright}
   }
   \caption{Realizability interpretation for \dlpaw}
   \label{fig:dlpaw_real}
   \end{figure*}

We finally recall the definition of the orthogonality relation w.r.t. a pole,
which is identical to the one for the $\lbvtstar$-calculus:
\begin{definition}[Orthogonality]
Given a pole $\pole$, we say that a {\cp} $\tis{p}{\tau}$ is {\em orthogonal} to a {\ce} $\tis{e}{\tau'}$
and write $\tis{p}{\tau}\orth\tis{e}{\tau'}$
if $\tau$ and $\tau'$ are compatible and $\cut{p}{e}\overline{\tau\tau'}\in\pole$.
The orthogonality between terms and co-terms is defined identically.
\end{definition}

We are now equipped to define the realizability interpretation of $\dlpaw$.
Firstly, in order to simplify the treatment of coinductive formulas, 
we extend the language of formulas with second-order variables $X,Y,\dots$ 
and we replace ${\nu^t_{fx} A}$ by $\nu^t_{Xx} A[X(y)/f(y)=0]$. 
The typing rule for co-fixpoint operators then becomes:
$$\infer[\cofixrule]{\sigmaopt\Gamma\sigdash \cofix{t}{bx}{p}: \nu^t_{Xx} A \optsigma}{
	\sigmaopt\Gamma\sigdash t:T  \optsigma
	& 
	\sigmaopt\Gamma,x:T,b:\forall y^T\!. X(y)\sigdash p:A  \optsigma
	& 
	X \notin\FV(\Gamma)
}$$
where $X$ has to be \text{ positive in } $A$.

Secondly, as in the interpretation of {\dltp} through Lepigre's calculus,
we introduce two new predicates, $p\in A$ for {\nef} proofs and $t\in T$ for terms.
This allows us to decompose the dependent products and sums into:
$$\hspace{-3.5mm}\begin{array}{c@{~}|@{\;}c}
\begin{array}{c@{~}c@{~}l}
\forall x^T\!.A &\defeq& \forall x.(x\in T \imp A)\qquad \\
\exists x^T\!.A &\defeq& \exists x.(x\in T \imp A)\qquad \\
\end{array}
&
\begin{array}{c@{~}c@{~}l}
\dptprod{a:A}{B}&\defeq& A \imp B \hfill(a\notin \FV(B))\\
\dptprod{a:A}{B}&\defeq& \forall a.(a\in A \imp B) \quad~(\text{otw.})\\
\end{array}
\end{array}
$$
This corresponds to the language of formulas and types defined by:
$$
\begin{array}{ccl}
T,U & ::= & \N \mid T \to U \mid t\in T \\
A,B  &::= &\top\mid \bot \mid X(t) \mid t = u\mid A\land B \mid A\lor B \\
&\mid& \forall x. A \mid \exists x. A \mid \forall a.A \mid \nu^t_{Xx}A \mid a\in A
\end{array}\leqno
\begin{array}{l}
\textbf{Types}\\
\textbf{Formulas}\\\\
\end{array}
$$
and to the following inference rules:
\setlength{\saut}{-0.5em}
$$
\begin{array}{c}
\infer[\autorule{\forall^a_r}]{\Gamma \sigdash v:\forall a.A}{\Gamma\sigdash v:A & a\notin \FV(\Gamma)}\qquad
\infer[\autorule{\forall^a_l}]{\Gamma \sigdash e:(\forall a.A)^\negt}{\Gamma\sigdash e:A[q/a] & q~\nef} 
\\\\[\saut]
\infer[\autorule{\forall^x_r}]{\Gamma \sigdash v:\forall x.A}{\Gamma\sigdash v:A & x\notin \FV(\Gamma)}\qquad
\infer[\autorule{\forall^x_l}]{\Gamma \sigdash e:(\forall x.A)^\negt}{\Gamma\sigdash e:A[t/x]}
\\\\[\saut]
\infer[\autorule{\exists^x_r}]{\Gamma \sigdash v:\exists x.A}{\Gamma\sigdash v:A[t/x]}\qquad
\infer[\autorule{\exists^x_l}]{\Gamma \sigdash e:(\exists x.A)^\negt}{\Gamma\sigdash e: A & x\notin \FV(\Gamma)}
\\\\[\saut]
\infer[\autorule{\in^p_r}]{\Gamma \sigdash p:p\in A}{\Gamma \sigdash p:A & p~\nef}           \qquad
\infer[\autorule{\in^p_l}]{\Gamma \sigdash e:(q\in A)^\negt}{\Gamma \sigdash e:A^\negt}     
\\\\[\saut]
\infer[\autorule{\in^t_r}]{\Gamma \sigdash t:t\in T}{\Gamma \sigdash t:T}                   \qquad
\infer[\autorule{\in^t_l}]{\Gamma \sigdash \pi:(t\in T)^\negt}{\Gamma \sigdash \pi:T^\negt}
\end{array}
$$
These rules are exactly the same as in Lepigre's calculus~\cite{Lepigre16} up to our stratified presentation in
a sequent calculus fashion, and modulo our syntactic restriction to {\nef} proofs instead of his semantical restriction.
It is a straightforward verification to check that the typability is maintained through the decomposition of dependent products and sums.

Another similarity with Lepigre's realizability model is that truth/falsity values 
will be closed under observational equivalence of proofs and terms. 
To this purpose, for each store $\tau$ we introduce the relation $\equiv_\tau$, which we define as 
the reflexive-transitive-symmetric closure of the relation $\typered_\tau$:
\begin{center}
\begin{tabular}{lcl}
   $t$ &$\typered_\tau$ &$t'\quad$ whenever $\quad \exists \tau',\forall \pi, (\cmd{t}{\pi}\tau\rightarrow\cmd{t'}{\pi}\tau'$ \\
   $p$ &$\typered_\tau$& $q \quad$ whenever $\quad \exists \tau',\forall f\,(\cmd{p}{f}\tau\rightarrow\cmd{q}{f}\tau')$ \\
\end{tabular}
\end{center}

All this being settled, it only remains to determine how to interpret coinductive formulas. 
While it would be natural to try to interpret them by fixpoints in the semantics,
this poses difficulties for the proof of adequacy\vshort{\footnote{See~\cite[Sec. 8.4.2]{these} for a discussion on this matter.}}.
\vlong{We discuss this matter in \Cref{app:cofix}, but as for now,
we will give a simpler interpretation.}
We stick to the intuition that since \texttt{cofix} operators are lazily evaluated,
they actually are realizers of every finite approximation of the (possibly infinite) coinductive formula.
Consider for instance the case of a stream: 
$$\texttt{str}^{0}_\infty {p}\defeq \cofix{0}{bx}{(p x,b(S(x)))}$$ of type $\nu^0_{Xx} A(x)\land X(S(x))$.
Such stream will produce on demand any tuple 
$(p0,(p1,...(pn,\square)...))$
where $\square$ denotes the fact that it could be any term, in particular $\texttt{str}_\infty^{n+1} {p}$.
Therefore, $\texttt{str}_\infty^0 p$ should be a successful defender of the formula 
$$(A(0)\land (A(1) \land ...(A(n)\land \top)...)$$
Since $\texttt{cofix}$ operators only reduce when they are bound to a variable in front of a forcing context,
it suggests interpreting the coinductive formula $\nu^0_{Xx} A(x)\land X(S(x))$ at level $f$ 
as the union of all the opponents to a finite approximation\vlong{~\footnote{See Appendix~\ref{app:cofix} for a discussion on this point.}}.

To this end, given a coinductive formula $\nu^0_{Xx} A$ where $X$ is positive in $A$, we define its finite approximations by:
$$F^0_{A,t} \defeq \top \qquad\qquad\qquad F^{n+1}_{A,t} \defeq A[t/x][F^{n}_{A,y}/X(y)]$$
Since $X$ is positive in $A$, we have for any integer $n$ and any term $t$
that $\fvf{F^{n}_{A,t}}\subseteq \fvf{F^{n+1}_{A,t}}$.
We can finally define the interpretation of coinductive formulas by:
$$\fvf{\nu^t_{Xx}A}   \defeq  \bigcup_{n\in\N}\fvf{F^n_{A,t}}$$

The realizability interpretation of closed formulas and types is defined in \Cref{fig:dlpaw_real}
by induction on the structure of formulas at level $f$, and by orthogonality at levels $V,e,p$.
When $S$ is a subset of $\P(\Lambda_p^\tau)$ (resp. $\P(\Lambda_e^\tau),\P(\Lambda_t^\tau),\P(\Lambda_\pi^\tau)$), 
we use the notation $S^{\pole_f}$ (resp. $S^{\pole_V}$, etc.)
to denote its orthogonal set restricted to $\Lambda_f^\tau$:
$$S^{\pole_f}\defeq \{\tis{f}{\tau}\in\Lambda^\tau_f : \forall \tis{p}{\tau'}\in S, \compat{\tau}{\tau'} \Rightarrow \cut{p}{f}\overline{\tau\tau'} \in\pole\}$$

At level $f$, closed formulas are interpreted by sets of strong forcing contexts-in-store $\tis{f}{\tau}$.
As explained earlier, these sets are besides closed under the relation $\equiv_\tau$ along their component $\tau$,
we thus denote them by $\P(\Lambda_f^\tau)_{/{\equiv_\tau}}$. 
Second-order variables $X,Y,\dots$ are then interpreted by functions from the set of terms $\Lambda_t$ to $\P(\Lambda_f^\tau)_{/{\equiv_\tau}}$
and as is usual in Krivine realizability~\cite{Krivine09}, for each such function $F$ we add a predicate symbol $\dot F$ in the language.

We shall now prove the adequacy of the interpretation with respect to the type system.
To this end, we need to recall a few definitions and lemmas. 
Since stores only contain proof terms, we need to define valuations for term variables in order 
to close formulas\footnote{Alternatively, we could have modified the small-step reduction rules to include substitutions of terms.}.
These valuations are defined by the usual grammar:
$$\rho ::= \varepsilon \mid \rho[x\mapsto V_t]\mid \rho[X\mapsto \dot F]$$
We denote by $\tis{p}{\tau}_\rho$ (resp. $p_\rho$, $A_\rho$) the proof-in-store $\tis{p}{\tau}$ where
all the variables $x\in\dom(\rho)$ (resp. $X\in\dom(\rho)$) have been substituted by the corresponding term $\rho(x)$ 
(resp. falsity value $\rho(x)$).

\begin{definition}Given a closed store $\tau$, a valuation $\rho$ and a fixed pole $\pole$, 
we say that the pair $(\tau,\rho)$ \emph{realizes} $\Gamma$, which we write\footnote{Once again, 
we should formally write $(\tau,\rho)\real_{\!\!\pole}\!\Gamma$ but we will omit the annotation by $\pole$ as often as possible.} $(\tau,\rho) \Vdash \Gamma$, if:
\begin{enumerate}
 \item for any $(a:A) \in\Gamma$, $\tis{a}{\tau}_\rho \in \tvV{A_\rho}$,
 \item for any $(\alpha:A_\rho^\negt) \in\Gamma$, $\tis{\alpha}{\tau}_\rho \in \fve{A_\rho}$,
 \item for any $\dpt{a|p}\in\sigma$, $a\equiv_{\tau} p$,
 \item for any $(x:T) \in\Gamma$, $x\in\dom(\rho)$ and $\tis{\rho(x)}{\tau} \in \tvVt{T_\rho}$.
 \end{enumerate}
 \label{def:store_real}
\end{definition}

%
%


We can check that the interpretation is indeed defined up to the relations $\equiv_\tau$:
\begin{proposition}\label{lm:dlpaw:equiv}
 For any store $\tau$ and any valuation $\rho$, the component along $\tau$ of the truth and falsity values defined in \Cref{fig:dlpaw_real}
 are closed under the relation $\equiv_\tau$:
 \begin{enumerate}
 \item if $\tis{f}{\tau}_\rho \in \fvf{A_\rho}$ and $A_\rho\equiv_\tau B_\rho$, then $\tis{f}{\tau}_\rho\in\fvf{B_\rho}$,
 \item if $\tis{V_t}{\tau}_\rho \in \tvVt{A_\rho}$ and $A_\rho\equiv_\tau B_\rho$, then $\tis{V_t}{\tau}_\rho\in\tvv{B_\rho}$.
\end{enumerate}
The same applies with $\tvp{A_\rho}$, $\fve{A_\rho}$, etc.
\end{proposition}

We can now prove the main property of our interpretation:

\begin{proposition}[Adequacy]\label{prop:dlpaw:adequacy}
  The typing rules are adequate with respect to the realizability interpretation,
  \emph{i.e.} typed proofs (resp. values, terms, contexts, etc.) belong to the corresponding truth values.
\end{proposition}
\begin{proof}
By induction on typing derivations such as given in the system
extended for the small-step reduction. 
\vlong{See Appendix~\ref{app:adequacy}.}\qedhere
\end{proof}

We can finally deduce that
$\dlpaw$ is normalizing and sound.
\begin{theorem}[Normalization]	
If~ $\Gamma \sigdash c$, then $c$ is normalizable. 
\end{theorem}
\begin{proof}
Direct consequence of \Cref{prop:dlpaw:norm_pole,prop:dlpaw:adequacy}.
\end{proof}


\begin{theorem}[Consistency]\label{thm:consistency}
$\nvdash_{\text{\dlpaw}}p: \bot$
\end{theorem}
\begin{proof}
 Assume there is such a proof $p$, by adequacy $\tis{p}{\varepsilon}$ is in $\tvp{\bot}$ for any pole.
 Yet, the set $\pole \defeq \emptyset$ is a valid pole, and with this pole, $\tvp{\bot}=\emptyset$, which is absurd.
\end{proof}

\section{Conclusion and perspectives}
\paragraph{Conclusion}
At the end of the day, we met our main objective, namely proving the soundness and the normalization 
of a language which includes proof terms for dependent and countable choice in a classical setting.
This language, which we called {\dlpaw}, provides us with the same
computational features as {\dpaw} but in a sequent calculus fashion.
These computational features allow {\dlpaw} to internalize the realizability approach of~\cite{BerBezCoq98,EscOli14} 
as a direct proofs-as-programs interpretation: both proof terms for countable and dependent choices 
furnish a lazy witness for the ideal choice function which is evaluated on demand. 
This interpretation is in line with the slogan that 
with new programing principles---here the lazy evaluation and the co-inductive objects---come
new reasoning principles---here the axioms $AC_\N$ and $DC$.

Interestingly, in our search for a proof of normalization for {\dlpaw}, we developed novel tools to study 
these side effects and dependent types in presence of classical logic.
On the one hand, we set out in \cite{Miquey17} the difficulties related to the definition of a sequent calculus with dependent types. 
On the other hand, building on ~\cite{MiqHer18}, we developed a variant of Krivine realizability adapted to
a lazy calculus where delayed substitutions are stored in an explicit environment.
The sound combination of both frameworks led us to the definition of {\dlpaw} together with its realizability interpretation.

\paragraph{
Krivine's interpretations of dependent choice}\label{s:dlpaw_quote}
The computational content we give to the axiom of dependent choice is pretty different of 
Krivine's usual realizer of the same~\cite{Krivine03}. 
Indeed, our proof uses dependent types to get witnesses of existential formulas, 
and we represent choice functions through the lazily evaluated stream of their values~\footnote{A similar
idea can be found in NuPrl BITT type theory, where choice sequences are used in place of functions~\cite{CohEtAl18}.}. 
In turn, Krivine realizes a statement which is logically equivalent to the axiom of dependent choice
thanks to the instruction \automath{\Quote}, which injectively associates a natural number 
to each closed $\lambda_c$-term.
In a more recent work~\cite{Krivine16}, Krivine proposes a realizability model which has a bar-recursor
and where the axiom of dependent choice is realized using the bar-recursion. 
This realizability model satisfies the continuum hypothesis and many more properties, 
in particular the real numbers have the same properties as in the ground model. 
However, the very structure of this model, where ${\bm \Lambda}$ is of cardinal $\aleph_1$ (in particular infinite streams of integer are terms),
makes it incompatible with $\Quote$.

It is clear that the three approaches are different in terms of programming languages. 
Nonetheless, it could be interesting to compare them from the point of view of the realizability models they give rise to.
In particular, our analysis of the interpretation of co-inductive formulas\vlong{\footnote{See also Appendix~\ref{app:cofix}.}} may suggest that
the interest of lazy co-fixpoints is precisely to approximate the limit situation where $\bm \Lambda$ has infinite objects.

\paragraph{Reduction of the consistency of classical arithmetic in 
  finite types with dependent choice to the consistency of
  second-order arithmetic}
The standard approach to the computational content of classical
dependent choice in the classical arithmetic in finite types is via
realizability as initiated by Spector~\cite{Spector62} in the context
of Gödel's functional interpretation, and later adapted to the context
of modified realizability by Berardi {\em et al}~\cite{BerBezCoq98}. 
The aforementioned works of Krivine~\cite{Krivine03,Krivine16} in the different settings of PA2 and ZF$_\varepsilon$ 
also give realizers of dependent choice.
In all these approaches, the correctness of the realizer, which
implies consistency of the system, is itself justified by a use at the
meta-level of a principle classically equivalent to dependent choice
(dependent choice itself in~\cite{Krivine03}, bar induction or update
induction~\cite{Berger04} in the case of \cite{Spector62,BerBezCoq98}.).

Our approach is here different, since we directly interpret proofs
of dependent choice in classical arithmetic computationally.
Besides, the structure of our realizability interpretation for {\dlpaw} suggests
the definition of a typed CPS to an extension of system $F$\vshort{~\cite{these}}\vlong{\footnote{See \cite[Chapitre 8]{these} for further details.}},
but it is not clear whether its consistency is itself conservative or not over system $F$.
Ultimately, we would be interested in a computational reduction 
of the consistency of {\dpaw} or {\dlpaw} to the one of PA2, that is 
to the consistency of second-order arithmetic.
While it is well-known that $DC$ is conservative over second-order arithmetic 
with full comprehension  (see \cite[Theorem VII.6.20]{Simpson09}),
it would nevertheless be very interesting to have such a direct computational reduction.
The converse direction has been recently studied by Valentin Blot, who presented in~\cite{Blot17} 
a translation of System F into a simply-typed total language with a variant of bar recursion.

%
\vspace{-0.3em}
\begin{acks}
 The author warmly thanks Hugo Herbelin for numerous discussions 
 and attentive reading of this work during his PhD years.
\end{acks}

\vspace{-0.3em}
\bibliographystyle{abbrvurl}
\vlong{\bibliography{biblio}}
\vshort{\bibliography{biblio-opt}}

\vlong{ 
\newpage
\appendix
\onecolumn
\emph{Since most of the proofs contain typing derivations, we switch to a one column format to ease their display.}

\section{Subject reduction (Proofs of Section~\ref{s:sub_red})}
\label{app:sub_red}

We detail here the proof of subject reduction for \dlpaw.
Recall that we define the relation $\sigma\dptimp \sigma'$
between lists of dependencies by:
 $$\sigma \dptimp \sigma' ~\defeq~\sigma(A) =\sigma(B) \limp \sigma'(A) =\sigma'(B) \eqno(\text{for any}~ A,B)$$
We first show that the cases which we encounter in the proof of subject reduction satisfy this relation:
\begin{lemma}[Dependencies implication]The following holds for any $\sigma,\sigma',\sigma''$:\vspace{-0.5em}
\begin{multicols}{2}
 \begin{enumerate}
 \item $\sigma\sigma'' \dptimp \sigma\sigma'\sigma'$
 \item $\sigma\dpt{(a_1,a_2)|(V_1,V_2)} \dptimp \sigma\dpt{a_1|V_1}\dpt{a_2|V_2}$
 \item $\sigma\dpt{\injec i a|\injec i V} \dptimp \sigma\dpt{a|V}$
 \item $\sigma\dpt{(x,a)|(t,V)} \dptimp \sigma\dpt{a|V}\dpt{x|t}$
 \item $\sigma\rdpt{(p_1,p_2)} \dptimp \sigma\dpt{a_1|p_1}\dpt{a_2|p_2}\rdpt{(a_1,a_2)}$
 \item $\sigma\rdpt{\injec i p} \dptimp \sigma\dpt{a|p}\rdpt{\injec i a}$
 \item $\sigma\rdpt{(t,p)} \dptimp \sigma\dpt{a|p}\rdpt{(t,a)}$
\end{enumerate}
\end{multicols}
\noindent where the fourth item abuse the definition of list of dependencies to include a substitution of terms.
\label{lm:dlpaw:dpt_imp}
\end{lemma}
\begin{proof}
 All the properties are trivial from the definition of the substitution $\sigma(A)$.
\end{proof}

We can now prove that the relation $\dptimp$ indeed matches the expected intuition:
\begin{repproposition}{prop:dlpaw:dpt_weak} [Dependencies weakening]
 If $\sigma,\sigma'$ are two dependencies list such that $\sigma\dptimp\sigma'$, then any derivation using $\sigma$ can be one using $\sigma'$ instead. 
 In other words, the following rules are admissible:
 $$
 \infer[\weakrule]{\Gamma \vdash^{\sigma'} J}{\Gamma \vdash^{\sigma} J}
 \qquad\qquad\qquad
 \infer[\weakdrule]{\Gamma \vdash_d J;\sigma'}{\Gamma \vdash_d J;\sigma}
 $$
\end{repproposition}
\begin{proof}
 Simple induction on the typing derivations. The rules {\resetrule} and {\cutrule} where the list of dependencies is used 
 exactly match the definition of $\dptimp$.
 Every other case is direct using the first item of \Cref{lm:dlpaw:dpt_imp}.
\end{proof}

In addition to the previous proposition, we need the following extra lemma to simplify the proof of subject reduction,
which we will use when concatenating two stores:
\begin{lemma}\label{lm:store_cat}
 The following rule is admissible:
 $$
 \infer[\taucatrule]{\Gamma\sigdash \tau_0\tau_1:(\Gamma_0,\Gamma_1;\sigma_0,\sigma_1)}{
    \Gamma\sigdash \tau_0:(\Gamma_0;\sigma_0)
    &
    \Gamma,\Gamma_0\vdash^{\sigma\sigma_0}\tau_1:(\Gamma_1;\sigma_1)
  }
  $$
\end{lemma}
\begin{proof}
 By induction on the structure of $\tau_1$.
\end{proof}

Finally, we need to prove that substitutions of terms are safe with respect to typing (recall that substitutions
of proofs are handled through the typing rules for stores):

\begin{lemma}\label{lm:safet}[Safe term substitution]
 If $~\Gamma\sigdash t:T $ then for any conclusion $J$ for typing proofs, contexts, terms, etc; the following holds:\nomidem
\begin{multicols}{2}
  \begin{enumerate}
  \item If ~$ \Gamma,x:T,\Gamma'\sigdash J       $ \quad~	then~~{$        \Gamma,\Gamma'[t/x]\vdash^{\sigma[t/x]} J[t/x]$.}
  \item If ~$ \Gamma,x:T,\Gamma'\vdash_d J;\sigma$ ~		then~~{$        \Gamma,\Gamma'[t/x]\vdash_d J[t/x];\sigma[t/x] $.}
 \end{enumerate}
 \end{multicols}\nomidem
\end{lemma}
\begin{proof}
 By induction on typing rules.
\end{proof}

We are now equipped to prove the expected property:
\begin{reptheorem}{thm:subject_reduction}
For any context $\Gamma$ and any closures $c\tau$ and $c'\tau'$ such that $c\tau \rightarrow c'\tau'$, we have:\nomidem
\begin{multicols}{2}
\begin{enumerate}
 \item If ~$\Gamma\vdash c\tau$~then~$\Gamma\vdash c'\tau'$.
 \item If ~$\Gamma\vdash_d c\tau;\varepsilon$~then~$\Gamma\vdash_d c'\tau';\varepsilon$.
 \end{enumerate}
\end{multicols}\nomidem
\end{reptheorem}
\nomidem
\begin{proof}
The proof follows the usual proof of subject reduction, by induction on the typing derivation and the reduction $c\tau \rightarrow c'\tau'$.
Since there is no substitution but for terms (proof terms and contexts being stored), there is no need for auxiliary lemmas 
about the safety of substitution. 
We sketch it by examining all the rules from \Cref{fig:dlpaw_types} from top to bottom.
~\\[0.2em]{\textbullet\quad} The cases for reductions of $\lambda$ are identical to the cases proven in the previous chapter for $\dltp$.
~\\[0.2em]{\textbullet\quad} The rules for reducing $\mu$ and $\tmu$ are almost the same
except that elements are stored, which makes it even easier.
For instance in the case of $\tmu$, the reduction rule is:
$$ \cmd{V}{\mut a.c\tau_1}\tau_0 				 \red c\tau_0[a:=V]\tau_1$$
A typing derivation in regular mode for the command on the left-hand side is of the shape:
$$
\infer[\lrule]{\Gamma\sigdash\cmd{V}{\mut a.c\tau_1}\tau_0}{
  \infer[\cutrule]{\Gamma,\Gamma_0\vdash^{\sigma\sigma_0}\cmd{V}{\mut a.c\tau_1}}{
    \infer{\Gamma,\Gamma_0\vdash^{\sigma\sigma_0}V:A}{\Pi_V}
    &
    \infer[\mutrule]{\Gamma,\Gamma_0\vdash^{\sigma\sigma_0}\mut a.c\tau_1:A^\negt}{
       \infer[\lrule]{\Gamma,\Gamma_0,a:A\vdash^{\sigma\sigma_0}c\tau_1}{
	  \infer{\Gamma,\Gamma_0,a:A,\Gamma_1\vdash^{\sigma\sigma_0\sigma_1}c}{\Pi_c}
	  &
	  \infer{\Gamma,\Gamma_0,a:A\vdash^{\sigma\sigma_0}\tau_1:(\Gamma_1;\sigma_1)}{\Pi_{\tau_1}}  
       }
    }
  }
  &
  \infer{\Gamma\vdash^{\sigma}\tau_0:(\Gamma_0;\sigma_0)}{\Pi_{\tau_0}}
}
$$
Thus we can type the command on the right-hand side:
{\small
$$
\infer[\lrule]{\Gamma\sigdash c\tau_0[a:=V]\tau_1}{
  \infer[\weakrule]{\Gamma,\Gamma_0,a:A,\Gamma_1\vdash^{\sigma\sigma_0\dpt{a|V}\sigma_1}c}{
    \infer{\Gamma,\Gamma_0,a:A,\Gamma_1\vdash^{\sigma\sigma_0\dpt{a|V}\sigma_1}c}{\Pi_c} 
  }
  &
  \infer[\taucatrule]{\Gamma\sigdash \tau_0[a:=V]\tau_1:(\Gamma_0,a:A,\Gamma_1;\sigma_0\dpt{a|V}\sigma_1)}{
    \infer[\tauprule]{\Gamma\sigdash \tau_0[a:=V]:(\Gamma_0,a:A;\sigma_0,\dpt{a|V})}{
      \infer{\Gamma\vdash^{\sigma}\tau_0:(\Gamma_0;\sigma_0)}{\Pi_{\tau_0}}
      &
      \infer{\Gamma,\Gamma_0\vdash^{\sigma\sigma_0}V:A}{\Pi_V}
    }
    &
    \infer{\Gamma,\Gamma_0,a:A\vdash^{\sigma\sigma_0}\tau_1:(\Gamma_1;\sigma_1)}{\Pi_{\tau_1}}
  }
}
$$
}
As for the dependent mode, the binding $\dpt{a|p}$ within the list of dependencies is compensated when typing the store as shown in the last derivation. 

~\\{\textbullet\quad} Similarly, elimination rules for contexts $\tmu[a_1.c_1|a_2.c_2]$, $\tmu(a_1,a_2).c$, $\tmu(x,a).c$ or $\muteq.c$ 
are easy to check, using \Cref{lm:dlpaw:dpt_imp} and the rule {\tauprule} in dependent mode to prove the safety with respect to dependencies.

~\\{\textbullet\quad} The cases for delimited continuations are identical to the corresponding cases for $\dltp$.

~\\{\textbullet\quad} The cases for the so-called ``call-by-value'' rules opening constructors are straightforward, 
using again \Cref{lm:dlpaw:dpt_imp} in dependent mode to prove the consistency with respect to the list of dependencies.

~\\{\textbullet\quad} The cases for the lazy rules are trivial.

~\\{\textbullet\quad} The first case in the ``lookup'' section is trivial. 
The three lefts correspond to the usual unfolding of inductive and co-inductive fixpoints. We only sketch the latter in regular mode.
The reduction rule is:
$$\cut{a}{f}\tau_0[a:=\cofix{t}{bx}{p}]\tau_1 ~\red~ \cmd{p[t/x][b'/b]}{\mut a.\cut{a}{f}\tau_1}\tau_0[b':=\lambda y.\cofix{y}{bx}{p}] $$
The crucial part of the derivation for the left-hand side command is the derivation for the cofix in the store:
{
$$
\infer[\tauprule]{\Gamma\sigdash \tau_0[a:=\cofix{t}{bx}{p}]:(\Gamma_0,a:\nu_{fx}^t A;\sigma_0)}{
  \infer{\Gamma\sigdash \tau_0:(\Gamma_0;\sigma_0)}{\Pi_{\tau_0}}
  &
  \infer[\cofixrule]{\Gamma,\Gamma_0\vdash^{\sigma\sigma_0} \cofix{t}{bx}{p}:\nu_{fx}^t A}{
    \infer{\Gamma\vdash^{\sigma\sigma_0} t:T}{\Pi_t}
    & 
    \infer{\Gamma,\Gamma_0, f:T\imp \N,x:T,b:\forall y^T. f(y)=0\vdash^{\sigma\sigma_0} p:A}{\Pi_p}
  }
}
$$
}
Then, using this derivation, we can type the store of the right-hand side command:
$$
\infer[\tauprule]{\Gamma\sigdash \tau_0[b':=\lambda y.\cofix{y}{bx}{p}]:\Gamma_0,b':-\forall y.\nu_{fx}^y A}{
 \infer{\Gamma\sigdash \tau_0:(\Gamma_0;\sigma_0)}{\Pi_{\tau_0}}
  &
  \infer[\farrule]{\Gamma,\Gamma_0\vdash^{\sigma\sigma_0} \lambda y.\cofix{y}{bx}{p}:\forall y.\nu_{fx}^t A}{
    \infer[\cofixrule]{\Gamma,\Gamma_0,y:T\vdash^{\sigma\sigma_0} \cofix{y}{bx}{p}:\nu_{fx}^y A}{
      \infer{\Gamma,\Gamma_0,y:T\vdash^{\sigma\sigma_0} y:T}{}
      & 
      \infer{\Gamma,\Gamma_0, f:T\imp \N,x:T,b:\forall y^T. f(y)=0\vdash^{\sigma\sigma_0} p:A}{\Pi_p}
    }
  }
}
$$
It only remains to type (we avoid the rest of the derivation, which is less interesting) the proof $p[t/x]$
with this new store to ensure us that the reduction is safe (since the variable $a$ will still be of  type $\nu^t_{fx} A$
when typing the rest of the command):
$$
  \infer[\convrrule]{\Gamma,\Gamma_0,b:\forall y.\nu_{fx}^y A\sigdash p[t/x]:\nu_{fx}^t A}{
    \infer{\Gamma,\Gamma_0,b:\forall y.\nu_{fx}^y A\sigdash p[t/x]:A[t/x][\nu_{fx}^y A/f(y)=0]}{\Pi_p}
    &
    \nu_{fx}^t A \equiv A[t/x][\nu_{fx}^y A/f(y)=0]
  }
$$
{
}

~\\{\textbullet\quad} The cases for reductions of terms are easy.
Since terms are reduced in place within proofs, the only things to check is that the reduction 
of $\wit$ preserves types (which is trivial) and that the $\beta$-reduction 
verifies the subject reduction (which is a well-known fact).  \qedhere

 \end{proof}
\newpage
 \section{Natural deduction as macros (Proofs of \Cref{s:macros})}
 \label{app:ded_nat}
 We give here two examples of typing rules for the macros $\subst p q$ and $\prf p q$ 
 (in natural deduction) that are admissible in {\dlpaw}.
 Recall that we have the following typing rules in {\dlpaw}:
 $$
 \infer[\!\!\exlrule]{\sigmaopt\Gamma \sigdash \mut (x,a).c: (\exists x^T. A)^\negt\optsigma }{\sigmaopt\Gamma, x:T, a:A\sigdash c\optsigma}
  \qquad
  \infer[\eqrule]{\sigmaopt\Gamma \sigdash \muteq.\cut{p}{e} : (t=u)^\negt\optsigma }{\sigmaopt\Gamma \sigdash p:A\optsigma & \sigmaopt\Gamma \sigdash e : A[u/t]\optsigma}	
  $$
  and that we defined $\prf p$ and $\subst p q$ as syntactic sugar: 
  $$\prf p \defeq \shift\cmd{p}{\tmu (x,a).\cmd{a}{\reset}}\qquad\qquad
  \subst p q\defeq \mu \alpha.\cmd{p}{\muteq.\cmd{q}{\alpha}}.$$
  Observe that $\prf p$ is now only definable if $p$ is a {\nef} proof term.
  For any $p\in\nef$ and any variables $a,\alpha$, we can prove the admissibility of the $\prfrule$-rule:
{

  $$\infer{\Gamma\sigdash \shift\cmd{p}{\tmu (x,a).\cmd{a}{\reset}}:A(\wit p)\mid \Delta}{
	\infer[\cutrule]{\cmd{p}{\tmu (x,a).\cmd{a}{\alpha}}:\Gamma\vdash_d \Delta,\reset:A(\wit p);\sigma\rdpt{p}}{
	  \Gamma\sigdash p:\exists x^{\N}.A\mid \Delta
	  \hspace{-0.7cm}
	  & 
	  \infer{\Gamma\mid \tmu (x,a).\cmd{a}{\reset}:\exists x^{\N}\!. A\vdash_d \Delta,\reset:A(\wit p);\sigma\rdpt{p}}{
	  \infer[\cutrule]{\cmd{a}{\alpha}:\Gamma,x:\N,a:A(x)\vdash_d \Delta,\reset:A(\wit p);\sigma\dpt{(x,a)|p}}{
	    \infer[\convrrule]{a:A(x)\sigdash a:A(\wit(x,a))}{
	      \infer{a:A(x)\sigdash a:A(x)}{}
	    }
	    ~~&~~
	    \infer[\resetrule]{\Gamma\mid \reset:A(\wit (x,a)) \vdash_d\reset:A(\wit p)\mid\Delta;\sigma}{
	      \sigma\dpt{(x,a)|p}(A(\wit p)=\sigma\dpt{(x,a)|p}(A(\wit (x,a)))}
	    }
	    }
	 }
      }
      $$
}  
Similarly, we can prove that the ($\subst{\!}{\!}$)-rule is admissible:
{

  $$\infer[\murule]{\Gamma\sigdash \mu \alpha.\cmd{p}{\muteq.\cmd{q}{\alpha}}:B[u]\mid \Delta}{
	\infer[\cutrule]{\cmd{p}{\muteq.\cmd{q}{\alpha}}:\Gamma\sigdash \Delta,\alpha:B[u]}{
	  \Gamma\sigdash p:t=u\mid \Delta
	  & \infer[\eqrule]{\Gamma\mid \muteq.\cmd{q}{\alpha}:t=u\sigdash \Delta,\alpha:B[u] }{	
	      \Gamma\sigdash q:B[t]\mid\Delta;\sigma
	      &
	      \infer[\axlrule]{\Gamma\mid \alpha:B[u] \sigdash\alpha:B[u]\mid\Delta}{}
	    }
	 }
      }.
      $$
}

 \begin{reptheorem}{thm:acn}[Countable choice~\cite{Herbelin12}]
 We have:
$$\begin{array}{@{~~}rc@{~~}l@{}l}
AC_\N		&:=& \lambda H.&\letop a=\cofix{0}{bn}{(H n, b(S(n))}\\
&&&\,\inop\, (\lambda n.\wit(\nthn n a),\lambda n.\prf(\nthn n a)\\
		&: & \multicolumn{2}{l}{\forall x^\N \exists y^T P(x,y) \imp \exists f^{\N\imp T} \forall x^\N P(x,f(x))}\\
\end{array}$$
where $\nthn n a:= \pi_1(\ind{n}{a}{x,c}{\pi_2(c)})$.
\end{reptheorem}
\begin{proof}
The complete typing derivation of the proof term for $AC_\N$ from Herbelin's paper~\cite{Herbelin12}
is given in Figure \ref{fig:ac_n}.
\end{proof}

\begin{reptheorem}{thm:dc}[Dependent choice~\cite{Herbelin12}]
 We have:
$$\begin{array}{r@{~~}c@{~~}l@{}l}
DC &:=& \lambda H.\lambda x_0.\letop~ a=(x_0,\cofix{0}{bn}{d_n})fsix~\\ 
   & & \inop~ (\lambda n.\wit(\nthn n a),(\refl,\lambda n.\pi_1(\prf(\prf(\nthn n a)))))\\
   &:& \forall x^T\!. \exists y^T\!. P(x,y) \imp \\ 
&& \qquad  \forall x_0^T\! .\exists  f\in T^\N.( f(0) = x_0 
 \land \forall n^\N\!. P( f (n), f (s (n))))\\
\end{array}$$
where $d_n:=\dest{H n}{y,c} (y,(c,b\,y)))$\\
and ~~~
$\nthn n a:= \ind{n}{a}{x,d}{(\wit(\prf d),\pi_2(\prf(\prf(d))))}$.
\end{reptheorem}
\begin{proof}
Left to the reader.
\end{proof}

 \begin{figure}[h]
\framebox{\vbox{
\input{ac_n}
}
}
\caption{Proof of the axiom of countable choice in \dlpaw}
\label{fig:ac_n}
\end{figure}
 
\newpage 

~\newpage 
\section{Small-step reduction rules (Proofs of \Cref{s:small_step})}
\label{app:small_step}
We give in Figure~\ref{fig:smallstep} the full reduction system based on small-step reduction rules which
are described in Section \ref{s:small_step}. 
We detail thereafter the proofs of the two main properties of the small-step reduction system.

\begin{repproposition}{p:small_sub_red}[Subject Reduction]
 The small-step reduction rules satisfy subject reduction.
\end{repproposition}
\begin{proof}
 The proof is again a tedious induction on the reduction $\reds$.
 There is almost nothing new in comparison with the cases for the big-step reduction rules: 
 the cases for reduction of terms are straightforward, as well as the administrative reductions changing the focus on a command.
 We only give the case for the reduction of pairs $(t,p)$.
 The reduction rule is:
 $$\cmdp{(t,p)}{e}\tau    \reds  \cmdp{p}{\coshift \cmd{t}{\mut x.\cmd{\coreset}{\tmu a.\cut{(x,a)}{e}}}}\tau 	$$
 Consider a typing derivation for the command on the left-hand side, which is of the shape (we omit the rule \lrule~and the store for conciseness):
 $$
   \infer[\cutrule]{\Gamma\sigdash \cmd{(t,p)}{e}}{
      \infer[\exrrule]{\Gamma\sigdash (t,p):\exists x^T\!.A}{
	\infer{\Gamma \sigdash t:T}{\Pi_t}
	&
	\infer{\Gamma \sigdash p:A[t/x]}{\Pi_p}
      }
      &
      \infer{\Gamma\sigdash e:(\exists x^T\!.A)^\negt}{\Pi_e}
   }
 $$
 Then we can type the command on the right-hand side with the following derivation:
 $$
 \infer[\!\!\cutrule]{\Gamma\sigdash \cmdp{p}{\coshift \cmd{t}{\mut x.\cmd{\coreset}{\tmu a.\cut{(x,a)}{e}}}}}{
    \Pi_p
    &\hspace{-0.5cm}
    \infer[\!\!\autorule{\tmu\coreset}]{\Gamma\sigdash \coshift \cmd{t}{\mut x.\cmd{\coreset}{\tmu a.\cut{(x,a)}{e}}}:A[t]^\negt}{
      \infer[\!\!\dcutrule]{\Gamma,\coreset:A[t]\vdash  \cmd{t}{\mut x.\cmd{\coreset}{\tmu a.\cut{(x,a)}{e}}};\sigma}{
	\Pi'_t
	&\hspace{-0.9cm}
	\infer[\!\!\autorule{\tmu_x}]{\Gamma,\coreset:A[t/x]\vdash_d \mut x.\cmd{\coreset}{\tmu a.\cut{(x,a)}{e}}:T;\sigma\rdpt{t}}{
	  \infer[\!\!\dcutrule]{\Gamma,\coreset:A[t],x:T\vdash_d \cmd{\coreset}{\tmu a.\cut{(x,a)}{e}};\sigma\dpt{x|t}}{
	    \infer[\!\!\mutrule]{\Gamma,x:T\sigdash \tmu a.\cut{(x,a)}{e}:A[x]^\negt}{
	      \infer[\!\!\cutrule]{\Gamma,x:T,a:A[x]\sigdash \cut{(x,a)}{e}:A[x]^\negt}{
		\Pi_{(x,a)} \qquad\qquad \Pi_e
	      }
	    }
	    &\hspace{-0.5cm}
	    A[t] = (\dpt{x|t})(A[x])
	  }
	}
      }
    }
}
$$
where $\Pi_{(x,a)}$ is as expected.
\end{proof}

\begin{repproposition}{prop:reduction}
 If a closure $c\tau$ normalizes for the reduction $\reds$, then it normalizes for $\red$.
\end{repproposition}
\begin{proof}
 By contraposition, one proves that if a command $c\tau$  produces an infinite 
 number of steps for the reduction $\red$, then it does not normalize for $\reds$ either.
 This is proved by showing by induction on the reduction $\red$ that each step, except for the contextual reduction of terms,
 is reflected in  at least on for the reduction $\reds$. 
 The rules for term reductions require a separate treatment, which is really not interesting at this point.  
 We claim that the reduction of terms, which are usual simply-typed $\lambda$-terms, is known to be normalizing anyway and does not deserve
 that we spend another page proving it in this particular setting.
\end{proof}

\begin{figure*}[h]
\framebox{\vbox{
\input{smallstep}
}
}
\caption{Small-step reduction rules}
\label{fig:smallstep}
\end{figure*}
 
 \newpage
 ~\newpage
 \section{Realizability interpretation (Proofs of \Cref{s:realizability})}
 \label{app:adequacy}
 We give here the different proofs relative to the realizability interpretation of \dlpaw.
 
 First, we verify that the set of normalizing closures is indeed a pole:
\begin{repproposition}{prop:dlpaw:norm_pole}
 The set $\pole_{\Downarrow}=\{c\tau\in\C_0:~c\tau\text{ normalizes }\}$ is a pole.
\end{repproposition}
\begin{proof}
 The first two conditions are already verified for the $\lbvtstar$-calculus~\cite{MiqHer18}. 
 The third one is straightforward, since if a closure $\cut{\shift c}{e}\tau$ is not normalizing, 
 it is easy to verify that $c[e/\reset]$ is not normalizing either. 
 Roughly, there is only two possible reduction steps for a command $\cut{\shift c}{e}\tau$:
 either it reduces to $\cut{\shift c'}{e}\tau'$, in which case $c[e/\reset]\tau$ also reduces
 to a closure which is almost $(c'\tau')[e/\reset]$;
 or $c$ is of the shape $\cut{p}{\reset}$ and it reduces to $c[e/\reset]\tau$. 
 In both cases, if $\cut{\shift c}{e}\tau$ can reduce, so can $c[e/\reset]\tau$. 
 The same reasoning allows us to show that if $c[V/\coreset]\tau$ normalizes,
 then so does $\cut{V}{\coshift c}\tau$ for any value $V$.
\end{proof}

We recall two key properties of the interpretation, whose proofs are similar to the proofs
for the corresponding statements in the $\lbvtstar$-calculus~\cite{MiqHer18}:
\begin{lemma}[Store weakening]
\label{lm:dlpaw:st_weak}
 Let $\tau$ and $\tau'$ be two stores such that $\tau\stext\tau'$, let $\Gamma$ be a typing context,
 let $\pole$ be a pole and $\rho$ a valuation. The following statements hold:
 \begin{enumerate}
  \item $\overline{\tau\tau'} = \tau'$
  \item If ~$\tis{p}{\tau}_\rho \in \tvp{A_\rho}$~ for some closed proof-in-store $\tis{p}{\tau}_\rho$ and formula $A$, then ~$\tis{p}{\tau'}_\rho\in\tvp{A_\rho}$. 
  The same holds for each level $e,E,V,f,t,\pi,V_t$ of the interpretation.
  \item If ~$(\tau,\rho)\real \Gamma$~ then ~$(\tau',\rho) \real \Gamma$.
 \end{enumerate}
\end{lemma}
\begin{proposition}[Monotonicity]
\label{prop:dlpaw:monotonicity}
For any closed formula $A$, any type $T$ and any given pole $\pole$, we have the following inclusions:
$$\tvV{A} \subseteq \tvp{A} \qquad\qquad \fvf{A}\subseteq \fve{A} \qquad\qquad \tvVt{T}\subseteq \tvt{T}$$
\end{proposition}

Truth and falsity values are defined up to observational equivalence:
 \begin{replemma}{lm:dlpaw:equiv}
 For any store $\tau$ and any valuation $\rho$, the component along $\tau$ of the truth and falsity values defined in \Cref{fig:dlpaw_real}
 are closed under the relation $\equiv_\tau$:
 \begin{enumerate}
 \item if $\tis{f}{\tau}_\rho \in \fvf{A_\rho}$ and $A_\rho\equiv_\tau B_\rho$, then $\tis{f}{\tau}_\rho\in\fvf{B_\rho}$,
 \item if $\tis{V_t}{\tau}_\rho \in \tvVt{A_\rho}$ and $A_\rho\equiv_\tau B_\rho$, then $\tis{V_t}{\tau}_\rho\in\tvv{B_\rho}$.
\end{enumerate}
The same applies with $\tvp{A_\rho}$, $\fve{A_\rho}$, etc.
\end{replemma}
\begin{proof}
 By induction on the structure of $A_\rho$ and the different levels of interpretation. 
 The different base cases ($p\in A_\rho$, $t\in T$, $t=u$) are direct since their components along $\tau$ are defined modulo $\equiv_\tau$, 
 the other cases are trivial inductions.
\end{proof}

 We can now give the complete proof of adequacy of the typing rules with respect to the realizability interpretation, defined in Figure~\ref{fig:dlpaw_real}.
 \begin{repproposition}{prop:dlpaw:adequacy}
  The typing rules are adequate with respect to the realizability interpretation.
 In other words, if~ $\Gamma$ is a typing context, $\pole$ a pole, $\rho$ a valuation and $\tau$ a store such that ${(\tau,\rho)\real \Gamma;\sigma}$,
then the following hold:
\begin{enumerate}
 \item If~ $v$ is a strong value 	such that $\Gamma\sigdash v:A$      ~or~~$\Gamma\vdash_d v:A      ;\sigma$, then $\tis{v}{\tau}_\rho   \in\tvV{A_\rho}$.
 \item If~ $f$ is a forcing context 	such that $\Gamma\sigdash f:A^\negt$~or~~$\Gamma\vdash_d f:A^\negt;\sigma$, then $\tis{f}{\tau}_\rho   \in\fvf{A_\rho}$.
 \item If~ $V$ is a weak value   	such that $\Gamma\sigdash V:A$      ~or~~$\Gamma\vdash_d V:A      ;\sigma$, then $\tis{V}{\tau}_\rho   \in\tvV{A_\rho}$.
 \item If~ $e$ is a context 		such that $\Gamma\sigdash e:A^\negt$~or~~$\Gamma\vdash_d e:A^\negt;\sigma$, then $\tis{e}{\tau}_\rho   \in\fve{A_\rho}$.
 \item If~ $p$ is a proof term  	such that $\Gamma\sigdash p:A$ 	    ~or~~$\Gamma\vdash_d p:A      ;\sigma$, then $\tis{p}{\tau}_\rho   \in\tvp{A_\rho}$.
 \item If~ $V_t$ is a term value   	such that $\Gamma\sigdash V_t:T$,	                                    then $\tis{V_t}{\tau}_\rho \in\tvVt{T_\rho}$.
 \item If~ $\pi$ is a term context   	such that $\Gamma\sigdash \pi:T$,	                                    then $\tis{\pi}{\tau}_\rho \in\fvpi{T_\rho}$.
 \item If~ $t$ is a term  		such that $\Gamma\sigdash t:T$, 	                                    then $\tis{t}{\tau}_\rho \in\tvt{T_\rho}$.
 \item If~ $\tau'$ is a store		such that $\Gamma\sigdash\tau':(\Gamma';)\sigma'$,                          then $(\tau\tau',\rho) \real (\Gamma,\Gamma';\sigma\sigma')$. 
 \item If~ $c$ is a command		such that $\Gamma\sigdash c$        ~or~~$\Gamma\vdash_d c        ;\sigma$, then $(c\tau)_\rho \in \pole$. 
  \item If~ $c\tau'$ is a closure	such that $\Gamma\sigdash c\tau'$   ~or~~$\Gamma\vdash_d c\tau'   ;\sigma$, then $(c\tau\tau')_\rho \in \pole$. 
\end{enumerate}
\end{repproposition}
\begin{proof}
The proof is done by induction on the typing derivation such as given in the system
extended with the small-step reduction $\reds$. 
Most of the cases correspond to the proof of adequacy for the interpretation of the $\lbvtstar$-calculus,
so that we only give the most interesting cases. To lighten the notations, we omit the annotation by the valuation $\rho$ whenever it is possible.
\prfcase{\exrrule}
We recall the typing rule through the decomposition of dependent sums:
$$\infer{\sigmaopt\Gamma \sigdash (t,p) : (u\in T \land A[u]) \optsigma}{\sigmaopt\Gamma \sigdash t:u\in T \optsigma & \sigmaopt\Gamma \sigdash p: A[u/x] \optsigma }$$
By induction hypothesis, we obtain that $\tis{t}{\tau}\in\tvt{u\in T}$ and $\tis{p}{\tau}\in\tvp{A[u]}$.
Consider thus any context-in-store $\tis{e}{\tau'}\in\fve{u\in T \land A[u]}$ such that $\tau$ and $\tau'$ are compatible, and let us denote by $\tau_0$ the union $\overline{\tau\tau'}$.
We have:
$$\cmdp{(t,p)}{e}\tau_0 \reds  \cmdp{p}{\coshift \cmd{t}{\mut x.\cmd{\coreset}{\tmu a.\cut{(x,a)}{e}}}}\tau_0$$ 
so that by anti-reduction, we need to show that $\coshift \cmd{t}{\mut x.\cmd{\coreset}{\tmu a.\cut{(x,a)}{e}}}\in \fve{A[u]}$.
Let us then consider a value-in-store $\tis{V}{\tau_0'}\in\tvV{A[u]}$ such that $\tau_0$ and $\tau_0'$ are compatible, and let us denote by $\tau_1$ the union $\overline{\tau_0\tau_0'}$.
By closure under delimited continuations, to show that $\cmdp{V}{\coshift \cmd{t}{\mut x.\cmd{\coreset}{\tmu a.\cut{(x,a)}{e}}}}\tau_1$ is in the pole 
it is enough to show that the closure $\cmd{t}{\mut x.\cmd{V}{\tmu a.\cut{(x,a)}{e}}}\tau_1$ is in $\pole$,. Thus it suffices to show that 
the coterm-in-store $\tis{\mut x.\cmd{V}{\tmu a.\cut{(x,a)}{e}}}{\tau_1}$ is in $\fvpi{u\in T}$.

Consider a term value-in-store $\tis{V_t}{\tau_1'}\in\tvVt{u\in T}$, such that $\tau_1$ and $\tau_1'$ are compatible, and let us denote by $\tau_2$ the union $\overline{\tau_1\tau_1'}$.
We have:
$$\cut{V_t}{\mut x.\cmd{V}{\tmu a.\cut{(x,a)}{e}}}{\tau_2}\reds \cmd{V}{\tmu a.\cut{(V_t,a)}{e}}\tau_2\reds \cut{(V_t,a)}{e}{\tau_2[a:=V]}$$ 
It is now easy to check that $\tis{(V_t,a)}{\tau_2[a:=V]}\in\tvV{u\in T \land A[u]}$ and to conclude, using \Cref{lm:dlpaw:st_weak} to get $\tis{e}{\tau_2[a:=V]}\in\fve{u\in T\land A[u]}$,
that this closure is finally in the pole.

\prfcase{\convrrule,\convlrule}
These cases are direct consequences of \Cref{lm:dlpaw:equiv} since if $A,B$ are two formulas such that $A\equiv B$,
in particular $A\equiv_\tau B$ and thus $\tvv{A} = \tvv{B}$.

\prfcase{\reflrule,\eqrule}
The case for $\refl$ is trivial, while it is trivial to show that $\tis{\muteq.\cut{p}{e}}{\tau}$ is in $\fvf{t=u}$
if $\tis{p}{\tau}\in\tvp{A[t]}$ and $\tis{e}{\tau}\in\fve{A[u]}$.
Indeed, either $t \equiv_{\tau} u$ and thus $A[t] \equiv_{\tau} A[u]$ (\Cref{lm:dlpaw:equiv},
or $t \nequiv_{\tau} u$ and $\fvf{t=u}=\Lambda_f^\tau$.

\prfcase{\autorule{\forall_r^x}} This case is standard in a call-by-value language with value restriction.
We recall the typing rule:
$$\infer[\autorule{\forall^x_r}]{\Gamma \sigdash v:\forall x.A}{\Gamma\sigdash v:A & x\notin \FV(\Gamma)}\qquad$$
The induction hypothesis gives us that $\tis{v}{\tau}_\rho$ is in $\tvV{A_\rho}$ for any valuation $\rho[x\mapsto t]$.
Then for any $t$, we have $\tis{v}{\tau}_\rho\in \fvf{A_\rho[t/x]}^{\pole_v} $
so that $\tis{v}{\tau}_\rho\in (\bigcap_{t\in\Lambda_t} \fvf{A[t/x]}^{\pole_v})$. 
Therefore if $\tis{f}{\tau'}_\rho$ belongs to  $\fvf{\forall x.A_\rho} = (\bigcap_{t\in\Lambda_t} \fvf{A[t/x]}^{\pole_v})^{\pole_f}$,
 we have by definition that $\tis{v}{\tau}_\rho\pole \tis{f}{\tau'}_\rho$.

\prfcase{\indrule}
We recall the typing rule:
$$\infer[\indrule]{\sigmaopt\Gamma\sigdash \ind{t}{p_0}{ax}{p_S} : A[t/x]\optsigma}{
	  \sigmaopt\Gamma\sigdash t: \N  \optsigma
	  &
	  \sigmaopt\Gamma \sigdash p_0:A [0/x]  \optsigma
	  &
	  \sigmaopt\Gamma,x:T,a:A\sigdash p_S:A[S(x)/x]\optsigma
	}
$$
We want to show that $\tis{\ind{t}{p_0}{ax}{p_S}}{\tau}\in\tvp{A[t]}$, let us then consider $\tis{e}{\tau'}\in\fve{A[t]}$
 such that $\tau$ and $\tau'$ are compatible, and let us denote by $\tau_0$ the union $\overline{\tau\tau'}$.
By induction hypothesis, we have\footnote{Recall that any term $t$ of type $T$ can be given the type $t\in T$.}
$t\in \tvt{t\in \N}$ and we have:
$$\cmdp{\ind{t}{p_0}{bx}{p_S}}{e}\tau_0\reds  \cmdp{\shift\cmd{t}{\mut y.\cmd{a}{\reset}[a:=\ind{y}{p_0}{bx}{p_S}]}}{e}\tau_0$$
so that by anti-reduction and closure under delimited continuations, it is enough to show that  the coterm-in-store
$\tis{\mut y.\cmd{a}{e}[a:=\ind{y}{p_0}{bx}{p_S}]}{\tau_0}$ is in $\fvpi{t\in \N}$.
Let us then consider $\tis{V_t}{\tau_0'}\in\tvVt{t\in \N}$
such that $\tau_0$ and $\tau_0'$ are compatible, and let us denote by $\tau_1$ the union $\overline{\tau_0\tau_0'}$.
By definition, $V_t=S^n(0)$ for some $n\in\N$ and $t\equiv_{\tau_1} S^n(0)$, and 
we have: 
$$\cmd{S^n(0)}{\mut y.\cmd{a}{e}[a:=\ind{y}{p_0}{bx}{p_S}]}\tau_1 \reds \cmd{a}{e}\tau_1[a:=\ind{S^n(0)}{p_0}{bx}{p_S}] $$
We conclude by showing by induction on the natural numbers that for any $n\in N$,
the value-in-store $\tis{a}{\tau_1[a:=\ind{S^n(0)}{p_0}{bx}{p_S}]}$ is in $\tvV{A[S^n(0)]}$.
Let us consider $\tis{f}{\tau_1'}\in\fvf{A[S^n(0)]}$ such that  the store
$\tau_1[a:=\ind{S^n(0)}{p_0}{bx}{p_S}]$ and $\tau_1'$ are compatible, 
and let us denote by $\tau_2[a:=\ind{S^n(0)}{p_0}{bx}{p_S}]\tau_2'$ their union.
\begin{itemize}
\item 
If $n=0$, we have:
$$
\cmd{a}{f}\tau_2[a:=\ind{0}{p_0}{bx}{p_S}]\tau_2' \reds \cmd{p_0}{\tmu a.\cut{a}{f}\tau'_2}\tau_2
$$
We conclude by anti-reduction and the induction hypothesis for $p_0$, since it is easy to show that $\tis{\tmu a.\cut{a}{f}\tau'_2}{\tau_2}\in\fve{A[0]}$.

\item 
If $n=S(m)$, we have:
$$
\cmd{a}{f}\tau_2[a:=\ind{S(S^m(0))}{p_0}{bx}{p_S}]\tau_2' \reds\cmdp{p_S[S^m(0)/x][b'/b]}{\mut a.\cmd{a}{f}\tau_2'}\tau_2[b':=\ind{S^m(0)}{p_0}{bx}{p_S}]
$$
Since we have by induction that $\tis{b'}{\tau_2[b':=\ind{S^m(0)}{p_0}{bx}{p_S}]}$ is in $\tvV{A[S^m(0)]}$,
we can conclude by anti-reduction, using the induction hypothesis for $p_S$
and the fact that $\tis{\tmu a.\cut{a}{f}\tau'_2}{\tau_2}$ belongs to $\fve{A[S(S^m(0))]}$.
\end{itemize}\noem

\prfcase{\cofixrule}
We recall the typing rule:
$$\infer[\cofixrule]{\sigmaopt\Gamma\sigdash \cofix{t}{bx}{p}: \nu^t_{Xx} A \optsigma}{
	\sigmaopt\Gamma\sigdash t:T  \optsigma
	& 
	\sigmaopt\Gamma,x:T,b:\forall y^T\!. X(y)\sigdash p:A  \optsigma
	& 
	X\text{ positive in } A & X\notin \FV(\Gamma)
}$$
We want to show that $\tis{\cofix{t}{bx}{p}}{\tau}\in\tvp{\nu^t_{Xx} A}$, let us then consider $\tis{e}{\tau'}\in\fve{\nu^t_{Xx} A}$
such that $\tau$ and $\tau'$ are compatible, and let us denote by $\tau_0$ the union $\overline{\tau\tau'}$.
By induction hypothesis, we have $t\in \tvt{t\in T}$ and we have:
$$\cmdp{\cofix{t}{bx}{p}}{e}\tau_0\reds  \cmdp{\shift\cmd{t}{\mut y.\cmd{a}{\reset}[a:=\cofix{y}{bx}{p}]}}{e}\tau_0$$
so that by anti-reduction and closure under delimited continuations, it is enough to show that the coterm-in-store
$\tis{\mut y.\cmd{a}{e}[a:=\cofix{y}{bx}{p}]}{\tau_0}$ is in $\fvpi{t\in \N}$.
Let us then consider $\tis{V_t}{\tau_0'}\in\tvVt{t\in T}$
such that $\tau_0$ and $\tau_0'$ are compatible, and let us denote by $\tau_2$ the union $\overline{\tau_0\tau_0'}$.
We have:
$$\cmd{V_t}{\mut y.\cmd{a}{e}[a:=\cofix{y}{bx}{p}]}\tau_1 \reds \cmd{a}{e}\tau_1[a:=\cofix{V_t}{bx}{p}] $$
It suffices to show now that the value-in store $\tis{a}{\tau_1[a:=\cofix{V_t}{bx}{p}]}$ is in $\tvV{\nu^{V_t}_{Xx}A}$.
By definition, we have:
$$\tvV{\nu^{V_t}_{Xx}A} = (\bigcup_{n\in\N}\fvf{F^n_{A,{V_t}}})^{\pole_V}  
= \bigcap_{n\in\N}\fvf{F^n_{A,{V_t}}}^{\pole_V} 
= \bigcap_{n\in\N}\tvV{F^n_{A,{V_t}}}$$
We conclude by showing by induction on the natural numbers that for any $n\in N$ and any $V_t$,
the value-in-store $\tis{a}{\tau_1[a:=\cofix{V_t}{bx}{p}]}$ is in $\tvV{F^n_{A,V_t}}$.

The case $n=0$ is trivial since $\tvV{F^0_{A,V_t}}=\tvV{\top}=\Lambda^\tau_V$.
Let then $n$ be an integer and any $V_t$ be a term value.
Let us consider $\tis{f}{\tau_1'}\in\fvf{F^{n+1}_{A,V_t}A}$ 
such that $\tau_1[a:=\cofix{V_t}{bx}{p}]$ and $\tau_1'$ are compatible, and let us denote by $\tau_2[a:=\cofix{V_t}{bx}{p}]\tau_2'$
their union.
By definition, we have: 
$$
\cmd{a}{f}\tau_2[a:=\cofix{V_t}{bx}{p}]\tau_2' \reds \cmd{p[V_t/x][b'/b]}{\tmu a.\cut{a}{f}\tau'_2}\tau_2[b':=\lambda y.\cofix{y}{bx}{p}]
$$
It is straightforward to check, using the induction hypothesis for $n$, that $\tis{b'}{\tau_2[b':=\lambda y.\cofix{y}{bx}{p}]}$ is in 
$\tvV{\forall y.y\in T\imp F^n_{A,y}}$.
Thus we deduce by induction hypothesis for $p$, denoting by $S$ the function $t\mapsto \fvf{F^n_{A,t}}$, that: 
$$\tis{p[V_t/x][b'/b]}{\tau_2[b':=\lambda y.\cofix{y}{bx}{p}]} \in \tvp{A[V_t/x][\dot S/X]} =\tvp{A[V_t/x][F^n_{A,y}/X(y)]} = \tvp{F^{n+1}_{A,V_t}}$$
It only remains to show that $\tis{\tmu a.\cut{a}{f}\tau'_2}{\tau_2}\in\fve{F^{n+1}_{A,V_t}}$, which is trivial from the hypothesis for $f$.
\end{proof}

\begin{reptheorem}{thm:consistency}[Consistency]
$\nvdash_{\text{\dlpaw}}p: \bot$
\end{reptheorem}
\begin{proof}
 Assume there is such a proof $p$, by adequacy $\tis{p}{\varepsilon}$ is in $\tvp{\bot}$ for any pole.
 Yet, the set $\pole \defeq \emptyset$ is a valid pole, and with this pole, $\tvp{\bot}=\emptyset$, which is absurd.
\end{proof}

\newpage
\section{About the interpretation of coinductive formulas}
\label{app:cofix}
While our realizability interpretation give us a proof of normalization and soundness for {\dlpaw},
it has two aspects that we should discuss. 
First, regarding the small-step reduction system, one could have expected the lowest level of interpretation to be $v$ instead of $f$.
Moreover, if we observe our definition, we notice that most of the cases of $\fvf{\cdot}$ are in fact defined by orthogonality to
a subset of strong values. 
Indeed, except for coinductive formulas, we could indeed have defined instead an interpretation  $\tvv{\cdot}$ of formulas at level $v$
and then the interpretation $\fvf{\cdot}$ by orthogonality:
$$\begin{array}{ccl}
     \tvv{\bot}	    & \defeq & \emptyset\\ 
     \tvv{t=u}	    & \defeq & \begin{cases}\refl & \text{if~} t\equiv u\\ \emptyset & \text{otherwise}\end{cases}\\
     \tvv{p\in A}   & \defeq & \{\tis{v}{\tau}\in\tvv{A}: v \eqtau p \}\\
     \tvv{T\imp B}  & \defeq & \{\tis{\lambda x .p}{\tau}  : \forall V_t \tau', \compat{\tau}{\tau'}\land \tis{V_t}{\tau'}\in\tvV{T} \Rightarrow \tis{p[V_t/x]}{\overline{\tau\tau'}}\in\tvp{B}\}\\
     \tvv{A\imp B}  & \defeq & \{\tis{\lambda a .p}{\tau}  : \forall V \tau', \compat{\tau}{\tau'}\land \tis{V}{\tau'}\in\tvV{A} \Rightarrow \tis{p}{\overline{\tau\tau'}[a:=V]}\in\tvp{B}\}\\
     \tvv{T\land A} & \defeq & \{\tis{(V_t,V)}{\tau}    : \tis{V_t}{\tau}\in\tvVt{T}  \land \tis{V}{\tau}   \in\tvV{A_2}\}\\
 \tvv{A_1\land A_2} & \defeq & \{\tis{(V_1,V_2)}{\tau}  : \tis{V_1}{\tau}\in\tvV{A_1} \land \tis{V_2}{\tau} \in\tvV{A_2}\}\\
 \tvv{A_1\lor A_2}  & \defeq & \{\tis{\injec i V}{\tau} : \tis{V}{\tau}\in\tvV{A_i}\}\\
 \tvv{\exists x.A}  & \defeq & \bigcup_{t\in\Lambda_t} \tvv{A[t/x]}\\
 \tvv{\forall x.A}  & \defeq & \bigcap_{t\in\Lambda_t} \tvv{A[t/x]}\\
 \tvv{\forall a.A}  & \defeq & \bigcap_{p\in\Lambda_p} \tvv{A[p/x]}\\
     \fvf{A} 	    & \defeq & \{\tis{f}{\tau} : \forall v \tau', \compat{\tau}{\tau'}\land \tis{v}{\tau'}\in\tvv{A} \Rightarrow \tis{v}{\tau'}\orth \tis{F}{\tau}\}\\
   \end{array}$$

If this definition is somewhat more natural, it poses a problem for the definition of  coinductive formulas. 
Indeed, there is a priori no strong value in the orthogonal of $\fvf{\nu^t_{fv} A}$, which is:
$$(\fvf{\nu^t_{fv} A})^{\pole_v} = (\bigcup_{n\in\N}\fvf{F^n_{A,t}})^{\pole_v} = \bigcap_{n\in\N}(\fvf{F^n_{A,t}})^{\pole_v})$$
For instance, consider again the case of a stream of type $\nu^0_{fx} A(x)\land f(S(x))=0$, a strong value in the intersection
should be in every $\tvv{A(0)\land (A(1)\land \dots (A(n)\land \top)\dots)}$, 
which is not possible due to the finiteness of terms\footnote{Yet, it might possible to consider interpretation with infinite proof terms, 
the proof of adequacy for proofs and contexts (which are finite) will still work exactly the same. 
However, another problem will arise for the adequacy of the \texttt{cofix} operator. Indeed, with the interpretation above,
we would obtain the inclusion:
\begin{center}$\bigcup_{n\in\N}(\fvf{F^n_{A,t}})\subset (\bigcap_{n\in\N}\tvt{F^n_{A,t}})^{\pole_f} = \fvf{\nu^t_{fx} A}$\end{center}
which is strict in general. By orthogonality, this gives us that
$\tvV{\nu^t_{fx} A}\subseteq {\bigcup_{n\in\N}(\fvf{F^n_{A,t}}))^{\pole_V}} $, 
while the proof of adequacy only proves that $\tis{a}{\tau[a:=\cofix{t}{b}{x}{p}]}$ belongs to the latter set.}.
Thus, the definition $\tvv{\nu^t_{fv} A}\defeq \bigcap_{n\in\N}\tvv{F^n_{A,t}}$ would give $\tvv{\nu^t_{fx}A} = \emptyset =\tvv{\bot}$.

Interestingly, and this is the second aspect that we shall discuss here, we could have defined instead 
the truth value of coinductive formulas directly by :
$$\tvv{\nu^t_{fx}A}\defeq \tvv{A[t/x][\nu^y_{fx} A/f(y)=0]}$$
Let us sketch the proof that such a definition is well-founded. 
We consider the language of formulas without coinductive formulas and extended with formulas of the shape $X(t)$ where $X,Y,...$ are parameters.
At level $v$, closed formulas are interpreted by sets of strong values-in-store $\tis{v}{\tau}$,
and as we already observed, these sets are besides closed under the relation $\equiv_\tau$ along their component $\tau$.
If $A(x)$ is a formula whose only free variable is $x$, the function which associates to each term $t$ the set $\tvv{A(t)}$ is thus
a function from $\Lambda_t$ to $\P(\Lambda_v^\tau)_{\equiv_\tau}$, let us denote the set of these functions by $\L$.
\begin{proposition}
 The set $\L$ is a complete lattice with respect to the order $\leq_\L$ defined by:\vspace{-0.5em}
 $$F\leq_\L G 	\defeq \forall t\in\Lambda_t. F(t)\subseteq G(t)$$
\end{proposition}
\begin{proof}
 Trivial since the order on functions is defined pointwise and the co-domain $\P(\Lambda^\tau_v)$ is itself a complete lattice.
\end{proof}

We define valuations, which we write $\rho$, as functions mapping each parameter $X$ to a function $\rho(X)\in\L$.
We then define the interpretations $\tvv{A}^\rho,\fvf{A}^\rho,...$ of formulas with parameters exactly as above with the additional
rule\footnote{Observe that this rule is exactly the same as in the previous section (see \Cref{fig:dlpaw_real}).}:
$$\tvv{X(t)}^\rho\defeq \{\tis{v}{\tau}\in\rho(X)(t)\}$$

Let us fix a formula $A$ which has one free variable $x$ and a parameter $X$ such that sub-formulas of the shape $X\,t$ only occur in positive positions in $A$.
\begin{lemma}\label{lm:compatibility}
Let $B(x)$ is a formula without parameters whose only free variable is $x$,
and let $\rho$ be a valuation which maps $X$ to the function $t\mapsto \tvv{B(t)}$.
Then $\tvv{A}^{\rho} = \tvv{A[B(t)/X(t)]}$
\end{lemma}
\begin{proof}
 By induction on the structure of $A$, all cases are trivial, and this is true for the basic case $A \equiv X(t)$:
 $$\tvv{X(t)}^\rho = \rho(X)(t) = \tvv{B(t)}\noem$$
\end{proof}

Let us now define $\varphi_A$ as the following function:
$$
\varphi_A:\left\{
  \begin{array}{ccc}
  \L& \to & \L \\
  F &\mapsto&t\mapsto \tvv{A[t/x]}^{[X\mapsto F]}
  \end{array}
\right.
$$

\begin{proposition}
 The function $\varphi_A$ is monotone.
\end{proposition}
\begin{proof}
 By induction on the structure of $A$, where $X$ can only occur in positive positions. 
 The case $\tvv{X(t)}$ is trivial, and it is easy to check that truth values are monotonic
 with respect to the interpretation of formulas in positive positions,
 while falsity values are anti-monotonic.
\end{proof}

We can thus apply Knaster-Tarski theorem to $\varphi_A$, and we denote by $\texttt{gfp}(\varphi_A)$ its greatest fixpoint.
We can now define:
$$\tvv{\nu^t_{Xx} A} \defeq \texttt{gfp}(\varphi_{ A})(t)$$
This definition satisfies the expected equality:
\begin{proposition}\label{prop:dlpaw:coind_def}
 We have:
 $$\tvv{\nu^t_{Xx} A} = \tvv{A[t/x][\nu^y_{Xx} A / X(y)]} $$
\end{proposition}
\begin{proof}
Observe first that by definition, the formula $B(z)=\tvv{\nu^z_{Xx} A}$ satisfies the hypotheses 
of \Cref{lm:compatibility} and that $\texttt{gfp}(\varphi_{ A}) = t\mapsto B(t)$.
Then we can deduce :
$$
   \tvv{\nu^t_{Xx} A} = \texttt{gfp}(\varphi_{ A})(t)
		      = \varphi_{ A}(\texttt{gfp}(\varphi_A))(t)
		      = \tvv{ A[t/x]}^{[X\mapsto \texttt{gfp}(\varphi_A)]}\\
		      =  \tvv{A[t/x][\nu^y_{Xx} A/X(y)]}\noem
$$
\end{proof}
Back to the original language, 
it only remains to define $\tvv{\nu^t_{fx} A}$ as the set $\tvv{\nu^t_{Xx} A[X(y)/f(y)=0]}$ that we just defined.
This concludes our proof that the interpretation of coinductive formulas
through the equation in \Cref{prop:dlpaw:coind_def} is well-founded.

We could also have done the same reasoning with the interpretation from the previous section, 
by defining $\L$ as the set of functions from $\Lambda_t$ to $\P(\Lambda_f^\tau)_{\equiv_\tau}$.
The function $\varphi_A$, which is again monotonic, is then:
$$
\varphi_A:\left\{
  \begin{array}{ccc}
  \L& \to & \L \\
  F &\mapsto&t\mapsto \tvv{A[t/x]}^{[X\mapsto F]}
  \end{array}
\right.
$$
We recognize here the definition of the formula $F^n_{A,t}$.
Defining $f^0$ as the function $t\mapsto\fvf{\top}$ and $f^{n+1} \defeq \varphi_A(f^n)$
we have:
$$\forall n\in\N, \fvf{F^n_{A,t}} = f^n(t) =\varphi_A^n(f^0)(t)$$

However, in both cases (defining primitively the interpretation at level $v$ or $f$),
this definition does not allow us to prove\footnote{To be honest, we should rather say that 
we could not manage to find a proof, and that we would welcome any suggestion from insightful readers.}
the adequacy of the {\cofixrule} rule.
In the case of an interpretation defined at level $f$, the best that we can do is to show 
that for any $n\in\N$, $f^n$ is a post-fixpoint since for any term $t$, we have:
$$f^{n}(t) = \fvf{F^{n}_{A,t}} \subseteq  \fvf{F^{n+1}_{A,t}} = f^{n+1}(t)=\varphi_A(f^n)(t) $$
With $\fvf{\nu^t_{fx} A}$ defined as the greatest fixpoint of $\varphi_A$, for any term $t$ and any $n\in\N$
we have the inclusion
$f^n(t) \subseteq \texttt{gfp}(\varphi_A)(t) = \fvf{\nu^t_{fx} A}$ and thus:
$$ \bigcup_{n\in\N} \fvf{F^{n}_{A,t}}  = \bigcup_{n\in\N} f^n(t)  \subseteq \fvf{\nu^t_{fx} A}$$
By orthogonality, we get: 
$$ \tvV{\nu^t_{fx} A}\subseteq \bigcap_{n\in\N} \tvV{F^{n}_{A,t}} $$
and thus our proof of adequacy from the last section is not enough to conclude that 
$\cofix{t}{bx}p \in \tvp{\nu^t_{fx} A}$.
For this, we would need to prove that the inclusion is an equality.
An alternative to this would be to show that the function $t\mapsto \bigcup_{n\in\N} \fvf{F^{n}_{A,t}}$
is a fixpoint for $\varphi_A$. In that case, we could stick to this definition and happily conclude that it
satisfies the equation:
$$\fvf{\nu^t_{Xx} A} = \fvf{A[t/x][\nu^y_{Xx} A / X(y)]} $$
This would be the case if the function $\varphi_A$ was Scott-continuous on $\L$ (which is a dcpo), 
since we could then apply Kleene fixed-point theorem\footnote{In fact, Cousot and Cousot proved a constructive version
of Kleene fixed-point theorem which states that without any continuity requirement, 
the transfinite sequence $(\varphi_A^\alpha(f^0))_{\alpha\in O_n}$ is stationary~\cite{Cousot79}. 
Yet, we doubt that the gain of the desired equality is worth a transfinite definition of the realizability interpretation.}
to prove that 
$t\mapsto \bigcup_{n\in\N} \fvf{F^{n}_{A,t}}$ is the stationary limit of $\varphi_A^n(f_0)$.
However, $\varphi_A$ is not Scott-continuous\footnote{In fact, this is nonetheless a good news about our interpretation. 
Indeed, it is well-know that the more ``regular'' a model is, the less interesting it is. 
For instance, Streicher showed that the realizability model induced by Scott domains (using it as a realizability structure) 
was not only a forcing model by also equivalent to the ground model. 
}
(the definition of falsity values involves double-orthogonal sets which 
do not preserve supremums), and this does not apply.

}

\end{document}